%% file: main.tex
\documentclass[
              ,11pt
              ,acmlarge
              ,nonacm
              ]{acmart}

\input{preamble}
\usepackage{paralist}

\usepackage{mathtools}
\usepackage{threeparttable}
\usepackage{booktabs}

\usepackage{tikz}
\usetikzlibrary{
  graphs,
  graphs.standard
}

\usepackage{enumitem}

\newcommand{\winner}{\rho_\text{winner}}

\DeclareMathOperator*{\argmax}{arg\,max}

\DeclareMathOperator{\Exec}{\mathsf{Exec}}
\newcommand{\msg}[1]{\langle #1 \rangle}

\newcommand{\Few}{\texttt{Few}}
\newcommand{\Undist}{\texttt{Undist}}

\newcommand{\Corr}{\texttt{Corr}}

\def\shay#1{}
\def\peter#1{}
\def\mingming#1{}
\def\xianbinZhu#1{}

\keywords{Distributed Algorithm; Leader Election; Lower Bound}

\author{Shay Kutten}
\affiliation{
  \institution{Technion - Israel Institute of Technology}
  \country{Israel}
}
\author{Peter Robinson}
\affiliation{
  \department{School of Computer \& Cyber Sciences}
  \institution{Augusta University}
  \country{Georgia, USA}
}
\author{Ming Ming Tan}
\affiliation{
  \department{School of Computer \& Cyber Sciences}
  \institution{Augusta University}
  \country{Georgia, USA}
}
\author{Xianbin Zhu}
\affiliation{
  \department{Department of Computer Science}
  \institution{City University of Hong Kong}
  \country{Hong Kong SAR}
}
\ccsdesc[500]{Theory of computation~Distributed algorithms}

\setboolean{short}{false} %
\onlyLong{
\tcolorboxenvironment{theorem}{lower separated=false,
                colback=white!90!gray,
colframe=white, fonttitle=\bfseries,
colbacktitle=white!50!gray,
coltitle=black,
enhanced,
boxed title style={colframe=black},
}
\tcolorboxenvironment{reptheorem}{lower separated=false,
                colback=white!90!gray,
colframe=white, fonttitle=\bfseries,
colbacktitle=white!50!gray,
coltitle=black,
enhanced,
boxed title style={colframe=black},
}
}

\begin{document}
\title{Improved Tradeoffs for Leader Election}
\begin{abstract}
  We consider leader election in clique networks, where $n$ nodes are connected by point-to-point communication links.
  For the \emph{synchronous clique under simultaneous wake-up}, i.e., where all nodes start executing the algorithm in round $1$, we show a tradeoff between the number of messages and the amount of time.
  The previous lower bound side of such a tradeoff, in the seminal paper of Afek and Gafni (1991), was shown only assuming adversarial wake-up.
  Interestingly, our new tradeoff also improves the previous lower bounds
for a large part of the spectrum, even under simultaneous wake-up.
  More specifically, we show that any deterministic algorithm with a message complexity of $n\,f(n)$ requires $\Omega\left(\frac{\log n}{\log f(n)+1}\right)$ rounds, for $f(n) > 1$.
  Our result holds even if the node IDs are chosen from a relatively small set of size $\Theta(n\log n)$, as we are able to avoid using Ramsey's theorem, in contrast to many existing lower bounds for deterministic algorithms.
  We also give an upper bound that improves over the previously-best tradeoff achieved by the algorithm of Afek and Gafni.
  Our second contribution for the synchronous clique under simultaneous wake-up is to show that $\Omega(n\log n)$ is in fact a lower bound on the message complexity that holds for any deterministic algorithm with a termination time $T(n)$ (i.e., any function of $n$), for a sufficiently large ID space. 
  We complement this result by giving a simple deterministic algorithm that achieves leader election in sublinear time while sending only $o(n \, \log n)$ messages, if the ID space is of at most linear size.
  We also show that Las Vegas algorithms (that never fail)
  require $\Theta(n)$ messages.
  This exhibits a gap between Las Vegas and Monte Carlo algorithms.

  For the \emph{synchronous clique under adversarial wake-up},
  we show that $\Omega(n^{3/2})$ is a lower bound for $2$-round algorithms. Our result is the first superlinear lower bound for \emph{randomized} leader election algorithms in the clique. We also give a simple algorithm that matches this bound.

  Finally, we turn our attention to the \emph{asynchronous clique}:
  Assuming \emph{adversarial wake-up}, we give a randomized algorithm that, for any $k\in[2,O(\log n / \log\log n)]$, achieves a message complexity of $O(n^{1 + 1/k})$ and an asynchronous time complexity of $k+8$.
  Our algorithm achieves the first tradeoff between messages and time in the asynchronous model.
  For \emph{simultaneous wake-up}, we translate the deterministic tradeoff algorithm of Afek and Gafni to the asynchronous model, thus partially answering an open problem they pose.
\end{abstract}
\acmYear{2023}\copyrightyear{2023}
\acmConference[PODC '23]{ACM Symposium on Principles of Distributed Computing}{June 19--23, 2023}{Orlando, FL, USA}
\acmBooktitle{ACM Symposium on Principles of Distributed Computing (PODC '23), June 19--23, 2023, Orlando, FL, USA}
\acmPrice{15.00}
\acmDOI{10.1145/3583668.3594576}
\acmISBN{979-8-4007-0121-4/23/06}
\maketitle

\section{Introduction}
\label{sec:intro}

\input{intro}

\input{preliminaries}

\section{The Synchronous Clique under Simultaneous Wake-up}
\label{sec:sync-clique}
\subsection{Communication Graphs} \label{sec:comm_graphs} \label{sec:model}

For deterministic algorithms, we assume that each node is equipped with a unique ID chosen from an \emph{ID universe} $U$, which is a ``sufficiently large'' set of integers.
An adversary chooses an $n$-subset from $U$, called \emph{(valid) ID assignment}, which, together with the port mapping, fully determines the execution of any given deterministic leader election algorithm (since the network is synchronous and all the nodes wake up simultaneously).
In some parts of our analysis, we are only interested in the behavior of the algorithm up to some specific round $r$, and we define the \emph{round $r$ execution prefix} to consist of the first $r$ rounds of the execution.

We say that two executions $E_1$ and $E_2$ are \emph{indistinguishable} for a node $u$ up to round $r$ if $u$ has the same ID in both executions and $u$ receives the exact same set of messages in each round until the start of round $r$.
A simple consequence of indistinguishability is that node $u$ will behave the same in both $E_1$ and $E_2$ and also output the same value in both executions, assuming $r$ is sufficiently large.
Our lower bounds in this section assume that $n$ is a power of 2; we do not need this restrictions for our algorithms, unless stated otherwise.

We need to reason about which parts of the network have communicated with each other, which is captured by the communication graph:

\begin{definition}[Communication Graph] \label{def:comm_graph}
Consider a round $r\ge 1$, an ID assignment $I$, and a partial port mapping $p$.
In the {round $r$ communication graph}, we have the same node set as the clique network, and there is a directed edge $(u,v)$, if node $u$ sent a message over a port that is connected to $v$ in some round $r'<r$.
The behavior of the algorithm is fully determined by $I$ and $p$, and hence we write $\mathcal{G}_r^{I,p}$ to denote the resulting \emph{round $r$ communication graph}.
\end{definition}
For brevity, we sometimes omit the superscripts $I$ and $p$  when they are irrelevant or clear from the context.
Note that $\mathcal{G}_1$ is the empty graph that consists only of singleton nodes without any edges.

Consider some weakly connected components $C_1,\dots,C_k \subseteq \mathcal{G}_r^{I,p}$, and let $S = \bigcup_{i=1}^k V(C_i)$, where $V(C_i)$ denotes the nodes in $C_i$.
By definition, the nodes in $S$ do not have any edges to nodes outside of $S$ in graph $\mathcal{G}_r^{I,p}$.
Therefore, the behavior of the nodes in $S$ up to round $r$ only depends on their $|S|$ IDs and port mapping and is, in particular, independent of the IDs assigned to the remaining nodes.
We say that \emph{$S$ is isolated up to round $r$ under $I$ and $p$}.
For the remainder of this section, whenever we use the term "component", we mean a weakly connected component. 

\begin{definition}[Capacity of Components] \label{def:capacity}
We say that a \emph{component $C$ has capacity $\lambda$} if, for each node $u \in C$, it holds that $u$ has no incoming and outgoing edge to at least $\lambda$ nodes in $C$.
\end{definition}
In other words, a capacity of $\lambda$ means that each node has at least $\lambda$ other nodes in $C$ to which it has not communicated yet. 

Since any deterministic algorithms must correctly work on all ID assignments and any port mapping between the node IDs, it is admissible for us to choose the mapping of the unused ports of nodes ``adaptively'', i.e., depending on the current state of the nodes.
The proof of the following lemma is immediate:

\begin{lemma} \label{lem:capacity}
  Let $C \subseteq \mathcal{G}_r^{I,p}$ be a component with capacity $\lambda$ and suppose that the nodes in $C$
  send (in total) at most $t\le\lambda$ messages $m_1,\dots,m_t$ during the interval of rounds $[r,r']$, for some round $r'\ge r$.
  Then, there exists a port mapping that is compatible with $p$ such that messages $m_1,\dots,m_t$ are received by nodes in $C$.
\end{lemma}

\begin{definition}[Restricted Execution Prefix] \label{def:prefix}
  Consider a set $X$ of at most $\frac{n}{2}$ IDs and suppose that we execute the given algorithm for exactly $r$ rounds on a set $V_X$ of $|X|$ nodes with IDs in $X$ that are connected according to some partial port mapping $p$.
  The resulting execution prefix will depend on the partial port mapping $p$, and we point out that the algorithm assumes that there are $n$ nodes and hence each node in $V_X$ has $n-1$ ports.
  We define the set $\Exec_r(X)$ to contain every round $r$ execution prefix of $V_X$,
  in which the port mapping $p$ is such that every message sent by a node $u \in X$ is received by some $u'\in X$.
\end{definition}
Intuitively speaking, $\Exec_r(X)$ consists of all of the execution prefixes where the nodes with IDs in $X$ operate ``undisturbed'', i.e., in isolation, from the rest of the network.
Notice that in every execution prefix $E \in \Exec_r(X)$, every node sends messages over less than $|X|$ distinct ports by round $r$, since otherwise, it would have sent a message to at least one node outside $X$ and thus $E \notin \Exec_r(X)$.

The next definition captures the scenario that nodes with IDs in a certain set $X$ can terminate ``on their own'', i.e., without sending messages to nodes with IDs not in $X$.
Consequently, there must exist some round $r$ after which the components formed by these nodes in $\mathcal{G}_r$ do not add any more outgoing edges.

\begin{definition}[Terminating and Expanding Components] \label{def:static}
Consider a set of IDs $X$ of size at most $n/2$. 
We say that \emph{$X$ forms terminating components} if there exists a round $r$ such that, for every execution prefix $E\in \Exec_r(X)$, it holds that all nodes with IDs in $X$ have terminated in $E$ by round $r$.
Otherwise, we say that \emph{$X$ forms expanding components}.
\end{definition}
We point out that when we say a set of IDs $X$ forms terminating components, it is for each partial port mapping that results in an execution prefix in $\Exec_r(X)$.
Throughout this section, all logarithms are assumed to be of base $2$ unless stated otherwise.
\onlyShort{In the full paper~\cite{DBLP:journals/corr/abs-2301-08235}, we show the following:}\xspace

\begin{lemma} \label{lem:number_static}
  Consider any ID universe $U'$ and any integer $\ell \le \log_2 n - 1$.
  There exist at most $2^{\log_2 n - \ell}$ disjoint subsets $X_i \in {U' \choose 2^\ell}$
  of size $2^\ell$ such that $X_i$ form terminating components\footnote{We will follow the standard notation from extremal combinatorics where ${U' \choose k}$ denotes the family of all the $k$-element subsets of $U'$.}.
\end{lemma}
\onlyLong{
\begin{proof}
  Fix some integer $\ell$ that satisfies the premise and assume towards a contradiction that there are at least $2^{\log_2 n - \ell}+1$ disjoint subsets $X_1,\dots X_{2^{\log_2 n - \ell}+1}$, each of which results in terminating components (see Definition~\ref{def:static}).

  That is, for each $i$, there exists a partial port mapping $p_i$ such that the nodes with IDs in $X_i$ only send messages to other nodes with IDs in $X_i$.
  Consider the set of IDs $X_a = \bigcup_{i=1}^{2^{\log_2 n - \ell}}X_i$.
  Since each $X_i$ has size exactly $2^\ell$, and these sets are pairwise disjoint, $X_a$ consists of $n$ unique IDs.
  Let $E_a$ be the execution of the algorithm on $X_a$ with the port mapping $p$ such that, for each node $u$ with ID in $X_i$ and all $k \in [n-1]$, we define $p(u,k) := p_i(u,k)$.
  By the correctness of the algorithm, this yields a leader with some ID $a \in X_a$.
  Without loss of generality, assume that $a \in X_1$.
  Now consider the set $X_b = \bigcup_{i=2}^{2^{\log_2 n - \ell}+1}X_i$.
  Since $X_b$ is disjoint union of $2^{\log_2 n - \ell}$ sets of size $2^\ell$ we know that $|X_b|=|X_a|$, and thus $X_b$ is also a valid ID assignment for $n$ nodes.
  Moreover, by combining the port mappings associated with $X_2,\dots,X_{2^{\log_2 n - \ell}+1}$ in a similar manner as for $E_a$, we obtain an execution $E_b$ and a leader with some ID $b$ when executing the algorithm on $X_b$.
  As we have assumed that each $X_i$ forms terminating components,
  it follows that, for each $u$ with ID in $X_i$ ($i \in [2,\log_2 n - \ell]$), some execution prefix in $\Exec_r(X_i)$ is indistinguishable from $E_a$ for $u$ as well as from $E_b$, for any round $r$.
  Thus, all the components formed by the sets $X_2,\dots,X_{\log_2 n - \ell}$ must behave the same in both executions $E_a$ and $E_b$ until the nodes have terminated, which implies that
  $b \in X_{2^{\log_2 n - \ell}+1}$.
  Now consider the execution $E$ on $X_1 \cup X_{2^{\log_2 n - \ell}+1} \cup
  \bigcup_{i=2}^{2^{\log_2 n - \ell}-1} X_i $ with port mapping $p$.
  It follows that we obtain leaders with IDs $a \in X_1$ and $b \in X_{2^{\log_2 n - \ell}+1}$ in $E$, yielding a contradiction.
\end{proof}
}

An immediate consequence of Lemma~\ref{lem:number_static} is that a sufficiently large ID universe must contain a large subset of IDs without any terminating component:

\begin{corollary} \label{cor:no_static}
  Consider any ID universe $U'$.
  There exists a set of IDs $U \subseteq U'$ of size $|U'| - n\log_2 n$ such that, for any $\ell \le \log_2 n - 1$ and any subset $X \subseteq U$ of size $2^\ell$, it holds that $X$ does not form terminating components.
\end{corollary}

\subsection{A Lower Bound on the Communication-Time Tradeoff in the Synchronous Clique} \label{sec:tradeoff}

\begin{theorem} \label{thm:tradeoff}
  Suppose that $n$ is a sufficiently large power of two.
  Consider any deterministic algorithm that elects a leader in the synchronous clique of $n$ nodes and terminates in $T(n)$ rounds while sending at most $n\cdot f(n)$ messages, where $f(n) > 1$ is any increasing function of $n$.
  If the IDs of the nodes are chosen from a set of size at least $2n\log_2 n + n$, then it must be that
    $T(n) >  \frac{\log_2 n - 1}{\log_2 f(n) + 1} + 1.$
  Consequently, any deterministic $k$-round algorithm requires at least $\Omega\lt( \lt(\frac{n}{2}\rt)^{1 + {1}/{(k-1)}}\rt)$ messages.
\end{theorem}

In the remainder of this section, we prove Theorem~\ref{thm:tradeoff}.
Let $U'$ be the given ID universe.
Corollary~\ref{cor:no_static} ensures that, after removing $n \log_2 n$ IDs from $U'$, the remaining IDs do not produce any terminating components of certain sizes (up to $n/2$).

Define $U \subseteq U'$ to be this set of the remaining IDs and observe that
\onlyLong{\begin{align*}
  |U| \ge  n\log_2 n + n.
\end{align*}}%
\onlyShort{$
  |U| \ge  n\log_2 n + n.
$}\xspace
Note that any execution of the algorithm on an ID assignment from $U$ results in a communication graph where the largest component contains a majority of the nodes upon the termination of the algorithm.

Assume towards contradiction that $T(n) \le  \frac{\log_2 n - 1}{\log_2 f(n) + 1} + 1.$
The following lemma is the key technical component for proving Theorem~\ref{thm:tradeoff}.
Intuitively speaking, it shows that by selectively excluding certain ID assignments, we can prevent components from growing too quickly.

\begin{lemma} \label{lem:num_components}
  Suppose that $n$ is a power of two.
  For every round $r \le T(n)$, there is a set $U_r \subseteq U$ of size at least $|U| - n\,(r-1)$ such that, for every ID assignment $I \subseteq U_r$, there exists a partial port mapping $p$, and a decomposition of $\mathcal{G}_{r}^{I,p}$ into sets of nodes $X_1,\dots,X_{n/2^{\sigma_r}}$,
  where
  \begin{align}
    \sigma_r := \lb(\lb\lceil \log_2 f(n) \rb\rceil + 1\rb)(r-1) \label{eq:sigma}
  \end{align}
  and the following properties hold for every set $X_i$:
  \begin{enumerate}
    \item[(A)] For some integer $k_I\ge 1$ (depending on $I$), there exist components $C_1,\dots,C_{k_I} \subseteq \mathcal{G}_r^{I,p}$ such that $X_i = \bigcup_{j=1}^{k_I} V(C_j)$, i.e., there are no edges between nodes in $X_i$ and nodes not in $X_i$.
    \item[(B)] $X_i$ contains exactly $2^{\sigma_r}$ nodes. %
  \end{enumerate}
\end{lemma}
\begin{proof}
  We proceed by induction on $r$.
  For the basis $r=1$, we define $U_1 := U$.
  Properties~(A) and (B) are immediate since $\mathcal{G}_1^{I,p}$ is the empty graph that consists only of singleton nodes without any edges, for every ID assignment $I \subseteq U_1$ and port mappings $p$.
  Moreover, $\sigma_1 = 0$ and hence $2^{\sigma_1} = 1$, as required.
  Now suppose that the lemma holds for some $r\ge 1$, i.e., there exists a set $U_r$ of size $|U| - n\,(r-1)$ such that all ID assignments from $U_r$ satisfy the inductive hypothesis with respect to $r$.
  We define
  \begin{align}
    \mu_{r+1} :=  2^{\sigma_r}\lb(2\, f(n)-1\rb). \label{eq:mu}
  \end{align}
  Consider any ID assignment $I \subseteq U_r$. Let $p$ be a partial port mapping and $X_1,\dots,X_{n/2^{\sigma_r}}$ be the decomposition of $\mathcal{G}_r^{I,p}$ induced by $I$ as guaranteed by the inductive hypothesis.
  We say that $I$ is \emph{costly}, if, when executing the algorithm on $I$ and $p$, there exists some $X_{I,p} \in \{X_1,\dots,X_{n/2^{\sigma_r}}\}$ such that the nodes in $X_{I,p}$ send at least $\mu_{r+1}$ messages during round $r+1$. We call the subset $X_{I,p}$ a \emph{costly subset} of $\mathcal{G}_{r}^{I,p}$. 
  We will iteratively prune the IDs of costly subsets from $U_r$ as follows:
  Let $I_1$ be any costly ID assignment from $U_r$. Then, there exists a costly subset $X_{I_1, p_1} \subset \mathcal{G}_{r}^{I_1,p_1}$ (where $p_1$ is the partial port mapping associated with $I_1$
according to the inductive hypothesis), which contains $2^{\sigma_r}$ nodes that send  (in total) at least $\mu_{r+1}$ messages in round $r+1$.
  We remove the IDs of $X_{I_1, p_1}$ from $U_r$ and obtain a slightly smaller ID universe denoted by $U_r^{(1)} \subseteq U_r$.
  Again, we check whether there exists an ID assignment $I_2 \subseteq U_r^{(1)}$ with a costly subset $X_{I_2, p_2}$ and proceed by removing the offending set of $2^{\sigma_r}$ IDs in the decomposition induced by $I_2$ from $U_r^{(1)}$, and so on.

  This pruning process stops after $\ell$ iterations, for some integer $\ell$, once we can no longer find any costly ID assignment in the ID universe $U_r^{(\ell)}$.
  We define $U_{r+1} := U_r^{(\ell)}$.
  \onlyShort{In the full paper~\cite{DBLP:journals/corr/abs-2301-08235}, we prove the following claim:}\xspace
  \begin{claim} \label{cl:costly}
    The pruning process stops after at most $\ell \le \frac{n}{2^{\sigma_r}}-1$ iterations, and
    \onlyLong{\begin{align*}
    |U_{r+1}| \ge |U_{r}| - n \ge |U| - n\,r.
    \end{align*}}%
    \onlyShort{
    $
    |U_{r+1}| \ge |U_{r}| - n \ge |U| - n\,r.
    $
    }\xspace
  \end{claim}
 \onlyLong{
  \begin{proof}[Proof of Claim~\ref{cl:costly}]
    By a slight abuse of notation, the same variable may refer to a set of nodes as well as to their IDs.
    Assume towards a contradiction that we remove at least $\ell := n / 2^{\sigma_r}$ sets $X_{I_1, p_1},\dots,X_{I_\ell, p_\ell}$,
    each of size $2^{\sigma_r}$, due to having identified some costly ID assignments $I_1,\dots,I_\ell$. Recall that $p_i$ is the partial port mapping associated with $I_i$ according to the inductive hypothesis of Lemma~\ref{lem:num_components}.
    
    According to the pruning process described above, we identify the next costly ID assignment $I_{i}$ after having removed the set $X_{I_{i-1}, p_{i-1}} \subseteq I_{i-1}$ of $2^{\sigma_r}$ IDs.
    It follows that the ID sets $X_{I_1, p_1},\dots,X_{I_\ell, p_\ell}$ are pairwise disjoint, which tells us that the set $J := \bigcup_{i=1}^{\ell} X_{I_i, p_i}$ is a valid ID assignment for $n$ nodes.
    By combining the port mappings $p_1,\dots,p_\ell$, we obtain a partial port mapping $p$ in a natural way: for $u \in X_{I_i, p_i}$, we define $p(u,k) := p_i(u,k)$.

    Let $E$ be the round $r+1$ execution prefix of the algorithm on $J$ and $p$, and let $E_{I_i}$ be the round $r+1$ execution prefix with IDs in $I_i$ and port mapping $p_i$.
    Since $X_{I_i, p_i}$ was one of the subsets of the decomposition induced by $I_i$, we know by the inductive hypothesis of Lemma~\ref{lem:num_components} that $X_{I_i, p_i}$ results in a set of components in execution $E_{I_i}$, i.e., there is no communication between the nodes with IDs in $X_{I_i, p_i}$ and the rest of the network.
    Consequently, $E_{I_i} \in \Exec_{r+1}(X_{I_i})$ and, by construction, $E$ and $E_{I_i}$ are indistinguishable for all nodes in $X_{I_i, p_i}$, for all $i$.
    It follows that the nodes in each set $X_{I_i, p_i}$ will also jointly send at least $\mu_{r+1}$ messages in round $r+1$ of $E$.
    Recalling \eqref{eq:mu}, this yields
    \begin{align*}
      \ell\cdot \mu_{r+1} = \frac{n}{2^{\sigma_r}} \lb( 2^{\sigma_r}\lb(2\, f(n)-1\rb)\rb) = 2\,n\,f(n)-n > n\,f(n)
    \end{align*}
    messages in $E$, where in the last inequality, we use the assumption that $f(n)>1$, contradicting the assumed bound on the message complexity of the algorithm.
  \end{proof}
  }

  Next, we describe the strategy of the adversary to ensure that, for any ID set chosen from $U_{r+1}$, the existing components do not expand too much.
  For any ID assignment $I \subseteq U_{r+1} \subseteq U_r$, consider the decomposition of $\mathcal{G}_{r}^{I, p}$ into $X_1,\dots,X_{n/2^{\sigma_r}}$ sets for a partial port mapping $p$ guaranteed by the inductive hypothesis.
  Since $I$ is not a costly ID assignment, we know that at most $\mu_{r+1}$, messages are sent by the nodes in any $X_i$ in round $r+1$.
  Any message that is sent over an already-used port by a node in $X_i$ will reach some other node in $X_i$ by construction, and thus we focus only on messages that are sent over previously unused ports in round $r+1$.

  Let $m_1\le \mu_{r+1}$ be the actual number of messages sent by nodes in $X_1$ during round $r+1$ over $m_1$ unused ports.
  We define
  \begin{align}
  t = 1 + \lceil \log_2 f(n) \rceil, \label{eq:t}
  \end{align}
  and connect these $m_1$ ports to arbitrary nodes in $\bigcup_{i=1}^{2^t} X_i$.
  We will argue below that there are sufficiently available ports in this set.
  Note that it is admissible for us to adaptively choose the port connections, as we are considering deterministic algorithms that works for all IDs assignments and, moreover, in the clean network model (see Section~\ref{sec:model}), a node $u \in X_1$ with an unused port $q$ does not know to which one of the nodes $q$ connects to until $u$ sends or receives a message across $q$.
  \onlyLong{
  In the following, we show that there are sufficiently many nodes in $\bigcup_{i=2}^{2^t} X_i$ to which the $m_1$ ports opened by nodes in $X_1$ can be connected to. Since $|X_i| = 2^{\sigma_r}$ by the inductive hypothesis, it follows from \eqref{eq:mu} that
  \begin{align}
    \left| \bigcup_{i=2}^{2^t} X_i \right|
    &= \sum_{i=2}^{2^t} \left| X_i \right| \notag \\
    & = (2^t-1) \, 2^{\sigma_r} \notag \\
    \ann{by \eqref{eq:t}} & \geq \lb(2\, f(n) - 1\rb) \, 2^{\sigma_r} \notag \\
    \ann{by \eqref{eq:mu}} & = \mu_{r+1}. \label{eq:ell}
  \end{align}
  }%
  \onlyShort{
  In the following, we show that there are sufficiently many nodes in $\bigcup_{i=2}^{2^t} X_i$ to which the $m_1$ ports opened by nodes in $X_1$ can be connected to. Since $|X_i| = 2^{\sigma_r}$ by the inductive hypothesis, it follows from \eqref{eq:mu} and \eqref{eq:t} that
  $
    \left| \bigcup_{i=2}^{2^t} X_i \right|
    = \sum_{i=2}^{2^t} \left| X_i \right|
     = (2^t-1) \, 2^{\sigma_r}
    \geq \lb(2\, f(n) - 1\rb) \, 2^{\sigma_r}
    = \mu_{r+1}.
  $
  }\xspace
  As there are no prior connections between the nodes in $X_1$ and $\bigcup_{i=2}^{2^t} X_i$ in $\mathcal{G}_{r}^{I, p}$ by assumption, the above bound tells us that there are sufficiently many nodes in $\bigcup_{i=2}^{2^t} X_i$ to which the $m_1$ ports opened by nodes in $X_1$ can be connected to.
  We define $X_1' = \bigcup_{i=1}^{2^t} X_i$, and we make this set $X_1'$ part of the decomposition of $\mathcal{G}_{r+1}^{I,p'}$ for some partial port mapping $p'$ that is compatible with $p$.
  Analogously, we can direct the $m_j \le \mu_{r+1}$ messages that are sent over previously unused ports by nodes in $X_j$, for $j \in [2,2^t]$, to nodes in $\bigcup_{i=1,i\ne j}^{2^t}X_i \subseteq X_1'$.

  We proceed similarly for the remaining sets:
  That is, we define $X_2' = \bigcup_{2^t+1}^{2^{2t}} X_i$ and connect the newly opened ports of the nodes in $X_{2^t+1}$ to nodes in $X_{2^t+2},\dots,X_{2^{2t}}$, and so forth.
  As a result, we obtain the required decomposition of $\mathcal{G}_{r+1}^{I,p'}$ into sets $X_1',\dots,X_{n/2^{t+\sigma_r}}'$.
  Note that it is possible to process all sets in this way since we assumed that $n$ is a power of $2$, and hence $n/2^t$ is an integer.

  To complete the proof, we need to show that each obtained $X_j'$ set satisfies (A) and (B).
  Property~(A) follows readily from the assumption that each set $X_i$ consists of components in $\mathcal{G}_{r}^{I,p}$ together with our strategy for connecting the newly opened ports in round $r+1$.
  In more detail, this ensures that a node in $X_j'$ does not have any edges to $\mathcal{G}_{r+1}^{I,p'} \setminus X_j'$.
  For Property~(B), we need to show that $|X_j'|$ contain exactly $2^{\sigma_{r+1}}$ nodes:
  According to the inductive hypothesis, $|X_i|=2^{\sigma_r}$ and, by the fact that the sets $X_i$ are pairwise disjoint, 
  for any $j$, we have that
  \begin{align}
  |X_j'|
    &= 2^{\sigma_r +t}  \notag\\
    \ann{by \eqref{eq:t}}
    &= 2^{\sigma_r + \lceil \log_2 f(n) \rceil + 1} \notag\\
    \ann{by \eqref{eq:sigma}}
    &= 2^{(\lceil \log_2 f(n) \rceil + 1)(r-1) + \lceil \log_2 f(n) \rceil + 1} \notag\\
    &= 2^{(\lceil \log_2 f(n) \rceil + 1)(r)} \notag\\
    \ann{by \eqref{eq:sigma}}
    &= 2^{\sigma_{r+1}}.\notag
  \end{align}
  This concludes the proof of Lemma~\ref{lem:num_components}.
\end{proof}

We now complete the proof of Theorem~\ref{thm:tradeoff}. Recall that the algorithm terminates by round $T(n)$ and that we assume towards contradiction that $T \le  \frac{\log_2 n - 1}{\log_2 f(n) + 1} + 1.$ Lemma~\ref{lem:num_components} implies that all components have size $2^{\sigma_T} \le 2^{(\log_2 f(n) +1)(T-1)} = 2^{\log_2 n -1} \le n/2$.
From Corollary~\ref{cor:no_static}, we know that the algorithm cannot terminate unless one of the components has a size of more than $n/2$. Therefore, we have arrived at a contradiction.

\subsection{An Improved Deterministic Algorithm}
\label{sec:improved_afekgafni}
\onlyShort{
In the full paper~\cite{DBLP:journals/corr/abs-2301-08235}, we describe a simple algorithm that can be viewed as an optimized variant of the deterministic synchronous algorithm in \cite{afek1991time} assuming simultaneous wake-up and prove the following result:
}%
\onlyLong{
We now describe a simple algorithm that can be viewed as an optimized variant of the deterministic synchronous algorithm in \cite{afek1991time} assuming simultaneous wake-up.

For a given integer parameter $k \ge 3$, the algorithm starts by executing $k-2$ iterations, each of which consists of two rounds.
Initially, every node is a \emph{survivor}.
In round $1$ of iteration $i$, each survivor sends its ID to $\lceil n^{\frac{i}{k-1}}\rceil$ other nodes that become \emph{referees}.
Then, in round $2$, each referee responds to the survivor with the highest ID and discards all other messages.
A node remains a survivor for iteration $i+1$ if and only if it received a response from every referee; otherwise it is \emph{eliminated}.
Finally, at the end of iteration $k-2$, we perform one iteration that consists only of the first round, in which all remaining survivors send a message to all other nodes (i.e., every node becomes a referee), and a survivor terminates as leader if its own ID is greater than all other IDs that it received.

To bound the remaining survivors at the end of the $i$-th iteration, observe that at most $\frac{n}{n^{i/(k-1)}} = n^{1 - i/(k-1)}$ of the survivors who sent out a message in this iteration could have received a response from all of their referees.
Thus, it follows that there remains only a single survivor at the end of the $(k-1)$-th iteration, who becomes leader.
Moreover, each survivor of iteration $i$ contacts $\lceil n^{(i+1)/(k-1)} \rceil$ referees in iteration $i+1$, which means that we send $O(n^{1 - i/(k-1) + (i+1)/(k-1)}) = O\lt( n^{1 + 1/(k-1)} \rt)$ messages per iteration, for a total message complexity of $O \lt( k\, n^{1 + 1/(k-1)} \rt)$.
Since every iteration except the final one require two rounds, the time complexity is $2(k-2) +1 = 2k - 3$.
In particular, for a time complexity of $\ell = 2k -3$ rounds, we obtain a message complexity of
\begin{align*}
O \lt( k\,n^{1 + 1/(k-1)} \rt)
  = O \lt( \ell\,n^{1 + \frac{1}{(\ell+3)/2-1}}  \rt)
  = O \lt( \ell\,n^{1 + \frac{2}{\ell+1}}  \rt).
\end{align*}
Thus we have shown the following:
}

\begin{theorem} \label{thm:improved_afekgafni}
  Consider the synchronous clique under simultaneous wake-up.
  For any odd integer $\ell \ge 3$, there exists a deterministic algorithm that terminates in $\ell$ rounds and sends $O \lt( \ell\,n^{1 + \frac{2}{\ell+1}} \rt)$ messages.
\end{theorem}
We point out that Theorem~\ref{thm:improved_afekgafni} improves over the algorithm of Afek and Gafni~\cite{afek1991time}, which, for $\ell$ rounds, achieves a message complexity of $O \lt( \ell\,n^{1 + 2/\ell} \rt)$.
In particular, for constant-time algorithms, we obtain a polynomial improvement in the message complexity.

\subsection{An Extreme Point of the Trade-off}

 \label{sec:unconditional}

By how much can we reduce the message complexity if we increase the time complexity?
We answer this question by showing
that any leader election algorithm must send at least $\Omega(n\log_2 n)$ messages, even if all the nodes wake up spontaneously at the same time.
This holds for all time-bounded algorithms where the termination time $T(n)$ is a function of $n$. %

\begin{theorem} \label{thm:lb_unconditional}
  Consider any deterministic algorithm that elects a leader in the synchronous clique in $T(n)$ rounds where the number of nodes $n$ is a sufficiently large power of two.
  If the IDs of the nodes are chosen from a set of size at least $n^{\log_2 n }(T(n))^{\log_2 n -1}$, then it must send at least $\Omega(n\log n)$ messages.
\end{theorem}

The main technical argument of our proof will focus on the restricted class of algorithms where each node sends at most one message per round, which we call \emph{single-send algorithms}.
The following lemma shows that single-send algorithms are equivalent to standard multicast algorithms in terms of their message complexity.

\begin{lemma} \label{lem:simulation}
  If there exists a multicast leader election algorithm that sends $M(n)$ messages and terminates in $T(n)$ rounds, then there exists a single-send leader election algorithm that sends at most $M(n)$ messages and terminates in at most $n\cdot T(n)$ rounds.
\end{lemma}
\onlyLong{
\begin{proof}
  We show how to simulate a given multicast algorithm $\mathcal{A}$ to obtain a new algorithm, called $\mathcal{S}$, that is a single-send algorithm.
  The idea behind the simulation is simple:
  We simulate each round $r \ge 1$ of the multicast algorithm $\mathcal{A}$, in an interval consisting of the sequence of rounds $(r-1)n+1,\dots,r\,n$ as follows:
  If, given a node $u$'s state at the start of round $r$, algorithm $\mathcal{A}$ requires $u$ to send messages $m_1,\dots,m_k$ during round $r$, the new algorithm $\mathcal{S}$ sends message $m_i$ in round $(r-1)n+i$.
  Since a node can send at most one message to any other node in a given round, we know that $k\le n-1$, which ensures that $u$ sends at most one message per round in algorithm $\mathcal{S}$.
  Moreover, any node $v$ that receives messages in some round $(r-1)n+j$, simply adds this message to a buffer and, at the end of round $r\, n$, it uses $\mathcal{A}$ to empty its buffer and process all messages accordingly.

  By a straightforward inductive argument, it follows that every node $u$ executing $\mathcal{S}$ processes the exact same set of messages at the end of round $r\, n$ that it receives in round $r$ when executing $\mathcal{A}$ and hence also performs the same state transitions.
  This implies the claimed time complexity of $\mathcal{S}$.
  The message complexity bound is immediate since $\mathcal{S}$ does not send any messages in addition to the ones produced by $\mathcal{A}$.
\end{proof}
}

The next lemma shows that the message complexity of leader election is high for single-send algorithms, which, together with Lemma~\ref{lem:simulation}, completes the proof of Theorem~\ref{thm:lb_unconditional}.
\onlyShort{
The proof of Lemma~\ref{lem:lb_single_send} is technically more involved and postponed to the full version~\cite{DBLP:journals/corr/abs-2301-08235}.
}%
\onlyLong{
The proof of Lemma~\ref{lem:lb_single_send} is technically more involved and postponed to Section~\ref{sec:proof_lb_single_send}.
}

\begin{lemma} \label{lem:lb_single_send}
  Consider a $T(n)$ time-bounded single-send algorithm $\mathcal{S}$ and an ID universe $U$ of size at least $n\,T(n)^{\log_2 n -1}$. %
  There exists an ID assignment $I \subseteq U$ such that $\mathcal{S}$ sends at least $\Omega(n \log n )$ messages under $I$.
\end{lemma}

\onlyLong{
\subsubsection{Proof of Lemma~\ref{lem:lb_single_send}} \label{sec:proof_lb_single_send}

\begin{lemma} \label{lem:growing}
  Consider an ID universe $U'$ of size at least $n\,T(n)^{\log_2 n -1}+n\, log(n)$ and let $U \subseteq U'$ be a set of IDs of the kind guaranteed  by Corollary~\ref{cor:no_static}.
  For every integer $i=0,\dots,\log_2 n - 1$, there exists a round
  $r_i$ and a set family $U_i \subseteq {U \choose 2^i}$ of $2^i$-element subsets of $U$ such that, for every set $I \in U_i$,
  every execution prefix $E$ in $\Exec_{r_i}(I)$, there exists a partial port mapping $p$ such that $E$
   forms a component $C$ of $2^i$ nodes in $\mathcal{G}_{r_i+1}^{\tilde{I},p}$, for all ID assignment $\tilde{I}$ that contains $I$, with the following properties:
  \begin{compactenum}
    \item[(a)] $C$ has capacity at least $\frac{|C|}{2}-1 = 2^{i-1}-1$ at the start of round $r_{i}+1$, if $i \geq 1$;
    \item[(b)] the nodes in $C$
    have opened $\frac{|C|}{2} = 2^{i-1}$
    previously unused ports during rounds $[r_{i-1}+1,r_i]$, if $i \ge 1$;
    \item[(c)]  $U_i$ contains at least $\frac{n}{2^i}T(n)^{\log_2 n -i -1}$ pairwise disjoint sets.
  \end{compactenum}
\end{lemma}
\begin{proof}
For the base case, we define $r_0 := 0$ and $U_0 := U$.
Note that the ``end of round $0$'' refers to the start of the (first) round $1$ of the algorithm, and hence $\mathcal{G}_1$ is the empty graph consisting of $n$ isolated vertices, each of which forms a singleton component.
Thus Properties~(a)-(c) hold.

For the inductive step, suppose that (a)-(c) hold for $i$, for some $0 \le i\le \log_2 n - 2$, which means that rounds $r_0,\dots,r_i$ are already defined.
Consider any $I \in U_i$ and  any execution $E_I \in \Exec_{r_i}(I)$. Let $C$ be the component formed in $E_I$ in round $r_i$.
By Property~(a), $C$ has a capacity of at least $\frac{|C|}{2}-1$ at the start of round $r_i+1$, and Lemma~\ref{lem:capacity} guarantees that it will be possible to connect the next $\frac{|C|}{2}-1$ previously unused ports that the nodes in $C$ open during rounds $[r_i+1,r_{i+1}]$ to other nodes within $C$.
Conceptually, opening these ports corresponds to adding $\frac{|C|}{2}-1$ directed edges to $C$ in the communication graph.
Since $2^i \le \frac{n}{2}$, Corollary~\ref{cor:no_static} tells us that $I$ does not form a terminating component and hence the nodes in $C$ cannot terminate without opening at least one port that is not connected to another node in $C$. %
It follows that there exists some earliest round $r_I$ such that, for all $r \in [r_i+1,r_I-1]$, the total number of newly opened ports used by the nodes in $C$ is
some $W\leq\frac{|C|}{2}-1$ and each node in some nonempty subset $S' \subseteq V(C)$
such that $|S'| \geq \frac{|C|}{2} -W $, opens one
new port in round $r_I$.
In other words, $r_I$ is the earliest round in which the nodes in $C$ have (collectively) opened at least $\frac{|C|}{2}$ ports since $r_i$.

Let $f : U_i \to [T(n)]$ be the mapping that assigns $r_I$ to each $2^i$-element subset $I \in U_i$, and define $g: [T(n)] \to [|U_i|]$ be the mapping such that $g(r)$ is equal to the number of $2^i$-element subsets $I \in U_i$ such that $f(I)=r$.
We define
\onlyLong{\begin{align*}
  r_{i+1} := \argmax_{r\le T(n)}\set{g(r)}.
\end{align*}}%
\onlyShort{
  $
  r_{i+1} := \argmax_{r\le T(n)}\set{g(r)}.
  $
}\xspace
Intuitively speaking, when considering all possible rounds in which the algorithm opens the $\frac{|C|}{2}$-th port since the start of round $r_i+1$, round $r_{i+1}$ is the one that occurs most frequently, breaking ties arbitrarily.
Define $U_i' \subseteq U_i$ to be the family of all $2^i$-subsets that open their $\frac{|C|}{2}$-th new port since round $r_i$ exactly in round $r_{i+1}$.

We now show how to construct isolated executions that yield components of size $2^{i+1}$ by merging the $2^i$ components in $U_i'$ in round $r_{i+1}$ as follows:
Let $U_{i+1}$ be the resulting family of $2^{i+1}$-element subsets obtained by taking the union of all possible pairs of disjoint $2^i$-subsets of $U_i'$.

For any $I \in U_i'$, let $E_I$ be an execution in $\Exec_{r_i}(I)$, and $S_I$ be the set of nodes that open a new port in round $r_{i+1}$ in $E_I$.
Now consider disjoint sets $I,I' \in U_i'$ and let $C$ and $C'$ be their components in $E_I$ and $E_{I'}$, respectively.
As we focus on single-send algorithms, we know that, in round $r_{i+1}$, at most $k \le 2^i$ nodes $u_1,\dots,u_k$
each open a new port in execution $E_I$ and, similarly, at most $k' \le 2^i$ nodes $v_1,\dots,v_{k'} \in S_{I'}$, each open a new port in execution $E_{I'}$.
Fix any one-to-one mappings $g : \set{u_1,\dots,u_k} \to I'$ and $h : \set{v_1,\dots,v_{k'}} \to I$.
We connect the port opened by $u_i$ to the node with ID $g(u_i)$ and, similarly, connect the port opened by $v_i$ to $h(v_i)$, yielding the (merged) component $D = C \cup C'$ of size $2^{i+1}$ at the end of round $r_{i+1}$ with IDs $I \cup I'$.
As a result, we have extended the $r_i$-round executions $E_I$ and $E_{I'}$ to an $r_{i+1}$-round execution in $\Exec_{r_{i+1}}
({I \cup I'})$.

To see why Property~(a) holds for  the merged component $D$ with IDs in $(I\cup I') \in U_{i+1}$, recall that we add at most $1$ incident edge per node in the communication graph between components $C$ and $C'$ in round $r_{i+1}$ and these are the first edges interconnecting $C$ and $C'$.
In particular, any node $u \in C$ has at least $2^{i}-1$ unused ports that can be connected to nodes in $C'$ and the same is true for nodes in $C'$.
This means that every node in $D$ must have at least
\onlyLong{\begin{align*}
|C|-1=|C'|-1=2^i-1 = \frac{|D|}{2}-1
\end{align*}}%
\onlyShort{
$
|C|-1=|C'|-1=2^i-1 = \frac{|D|}{2}-1
$}\xspace
unused ports, which proves the claimed capacity bound for the merged component $D$.

Property~(b) follows since the nodes in $C$ as well as the ones in $C^{'}$ have opened $\frac{|C|}{2}$ new ports during $[r_i+1,r_{i+1}]$ according to our port mapping described above.

Finally, to show that Property ~(c) holds for $U_{i+1}$, a counting argument reveals that
\onlyLong{\begin{align*}
	|U_i'| \ge \frac{|U_i|}{T(n)}.
\end{align*}}%
\onlyShort{
$
	|U_i'| \ge \frac{|U_i|}{T(n)}.
$
}\xspace
Since $U_i$ contains at least $\frac{n}{2^i}T(n)^{\log_2 n -i -1}$ pairwise disjoint sets by induction, the lower bound on $|U_i'|$ shows that $U_i'$ contains at least $\frac{n}{2^i}T(n)^{\log_2 n -i -2} $ pairwise disjoint sets.
Note that any set of four pairwise disjoint subsets $X_1, X_2, X_3, X_4$ in $U_i'$ will give a pair of disjoint subsets $X_1 \cup X_2, X_3 \cup X_4$ in $U_{i+1}$.
Hence, $U_{i+1}$ contains at least $\frac{n}{2^{i+1}}T(n)^{\log_2 n -i -2} $ pairwise disjoint sets of size $2^{i+1}$.
\end{proof}

We now complete the proof of Lemma~\ref{lem:lb_single_send}:
According to Lemma~\ref{lem:growing}, the family $U_{\log_2 n - 1}$ contains at least $2$ disjoint sets $I'$ and $I''$, each of size $\frac{n}{2}$.
Consequently, $I^* = I' \cup I''$ is a valid ID assignment.
By construction, $I^*$ consists of the union of two sets in $U_{\log_2 n - 1}$, which in turn each consist of the union of sets in $U_{\log_2 n - 2}$ and so forth.
This means that, for every round $r_j$ ($j \in [0,\log_2 n - 1]$) guaranteed by Lemma~\ref{lem:growing}, there exist $\frac{n}{2^j}$ sets $I_1^*,\dots,I^*_{n/2^j}$ of size $2^j$ such that $I^* = \bigcup_{k=1}^{n/2^j} I_k^*$, and, for all $k$, any restricted execution prefix in $\Exec_{r_j}(I_k^*)$ will satisfy Property~(b). %
Now, for any restricted execution prefix $E_{I_1^*}, \ldots, E_{I_{n/{2^j}}^*}$ in
$\Exec(I_1^*),\dots,\Exec(I_{n/{2^j}}^*)$ respectively,
consider the $n$-node execution $E^*$ obtained by running the executions $E_{I_1^*}, \ldots, E_{I_{n/{2^j}}^*}$ in parallel up to round $r_{j}$.
By construction,
there will be no messages sent between the components and hence a node $u$ with some ID $x \in I_k^*$ will behave exactly in the same way as it does in $E_{I_k^*}$, because $E_{I_k^*}$ and $E^*$ are indistinguishable for $u$ until round $r_j$.
Consequently, we obtain $n/2^j$ components $C_1,\dots,C_{n/2^j}$ at the end of round $r_j$, and Property~(b) is guaranteed to hold for each of these components even in execution $E^*$.
This shows that $2^{j-1} \frac{n}{2^j} = \frac{n}{2}$ new ports are opened in the interval $[r_{j-1}+1,r_j]$, for every $j\ge 1$.
Summing up over all $j=1,\dots,\log_2 n - 1$, it follows that $\Omega(n \log n)$ ports are opened in total throughout execution $E^*$.
}

\subsubsection{An Algorithm with $o(n\log_2 n)$ Message Complexity and Sublinear Time for Small ID Universes} \label{sec:algo}

\onlyShort{
In the full paper~\cite{DBLP:journals/corr/abs-2301-08235}, we show that requiring a sufficiently large ID space in Theorem~\ref{thm:lb_unconditional} is indeed necessary, by giving a deterministic algorithm that, for an ID universe of size $O(n)$, elects a leader in sublinear time and sends $o(n \log_2 n)$ messages.
}
\onlyLong{
We now show that requiring a sufficiently large ID space in Theorem~\ref{thm:lb_unconditional} is indeed necessary, by giving a deterministic algorithm for the case where the range of IDs of the nodes is restricted to the set $\{1,\dots,n\,g(n)\}$, where $g(n) \ge 1$ is any integer-valued function of $n$.
In the trivial case where the IDs are from $1$ to $n$, then the algorithm that chooses $1$ to be the ID of the leader is optimal.
In the following algorithm, we  generalize this observation.
We use $id(u)$ to denote the ID of a node $u$. Parameter $d\le n$ controls the trade-off between the time and message complexity.

\begin{algorithm}
\begin{algorithmic}[1]
	 \For{each $u \in V$ execute the following steps in parallel}
	     \For{round $i = 1$ to $i=\lceil \frac{n}{d} \rceil$}
	         \If{ $id(u) \in [(i-1) \, d \, g(n)+1, i \, d \, g(n) ] $}
	           \State Node $u$ sends its ID to all other nodes.
           \EndIf
           \State If a node receives messages in round $i$, then it selects the smallest ID from all received IDs in round $i$ (including its own value, if it sent one) as the leader.
      \EndFor
    \EndFor
\end{algorithmic}
\caption{Deterministic Algorithm for Small ID Universes}\label{alg:g}
\end{algorithm}

\begin{theorem}
    Suppose that nodes are chosen from the set $\set{1,\dots,n\,g(n)}$, where $g(n)\ge 1$ is any integer-valued function of $n$. For any $d \leq n$, Algorithm \ref{alg:g} outputs a leader within $\lceil \frac{n}{d} \rceil$ rounds with message complexity $n \, d \, g(n)$.
    In particular, for an ID universe of size $O(n)$, we obtain a sublinear time algorithm that sends $o(n \log n)$ messages.
\end{theorem}
}

\onlyLong{
\begin{proof}
	Let $u$ be the node with the minimum ID and $i^*$ is such
	that $id(u) \in [ (i^*-1) \, d \, g(n)+1, i^*\, d \, g(n)]$.
	Then at round $i^*$, each node with ID in
	$[ (i^*-1) \, d \, g(n)+1, i^*\, d \, g(n)]$
	will send its ID to all other nodes.
	At the end of round $i^*$, all nodes will receive the same set of IDs and select node $u$ as the leader.

	Since all nodes' IDs are less than or equal to $n\,g(n)$,
 $i^*$ is at most $\lceil \frac{n}{d} \rceil$. Hence, within $\lceil \frac{n}{d} \rceil$ rounds, the algorithm terminates.
	Moreover, there are at most $d \, g(n)$ nodes with IDs in $[ (i^*-1) \, d \, g(n)+1, i^*\, d \, g(n)]$, hence, there are at most $d \, g(n)$ nodes that send their IDs to all other nodes at round $i^*$, giving the message complexity of $n \, d \, g(n)$. %
    In particular, when $g(n)$ is a constant and $d$ is $o(\log(n))$, then the message complexity is $o(n\log(n))$.
\end{proof}
}

\subsection{Las Vegas Algorithms} \label{sec:las_vegas}

As elaborated in more detail in Section~\ref{sec:intro}, Kutten et al.~\cite{kutten2015sublinear} give a randomized Monte Carlo algorithm that successfully solves leader election with high probability using only $O(\sqrt{n}\log ^{3/2}n)$ messages while terminating in just $2$ rounds.
Their algorithm is for implicit leader election (that is, not every node needs to know the name of the leader eventually \cite{lynch1996distributed}), and is a Monte Carlo one, i.e., fails with small probability of error.
One may ask whether those two assumptions were necessary. First, we comment that it is trivial to turn their Monte Carlo algorithm into a Las Vegas explicit one. (Their algorithm elects a leader in the second round; one can add a step where the
leader announces its identity to everybody at the third round: a node who does not receive an announcement in the third round restarts the algorithm.) Unfortunately, this obvious transformation of their algorithm increases the message complexity to $O(n)$.
\onlyLong{Below, we show that this is necessary, even for implicit leader election, i.e., if the nodes do not need to learn the identity of the leader.}%
\onlyShort{In the full paper~\cite{DBLP:journals/corr/abs-2301-08235}, we give a simple indistinguishability argument to prove the following:}

\begin{theorem} \label{thm:las_vegas}
  Any randomized Las Vegas leader election algorithm requires $\Omega(n)$ messages in expectation.
  Moreover, there exists an algorithm that with high probability, achieves $O(n)$ messages and terminates in $3$ rounds.
\end{theorem}
\onlyLong{
\begin{proof}
  The upper bound follows immediately from the previous discussion.

  For the lower bound, suppose that there exists a Las Vegas randomized algorithm that solves implicit leader election, i.e., exactly one node becomes leader whereas everyone else becomes non-leader, with an expected message complexity of $o(n)$. %
  Let $E$ be the event that the algorithm sends $o(n)$ messages, and notice that $E$ is guaranteed to happen with constant probability, for every  assignment.
  Now consider some ID assignment $I$ on a network with an even number of nodes.
  Conditioned on $E$, there exists a set of $n/2$ nodes with IDs $S \subseteq I$ such that no node in this set receives or sends any message from any other node throughout the execution.
  By the correctness of the algorithm, all nodes in $S$ must eventually terminate.
  Let $L_S$ be the indicator random variable that is $1$ if and only if some node in $S$ becomes leader, i.e., $L_S=0$ if none of them does.
  Clearly, $\Pr\lt[ L_S \!=\! 1 \mid E \rt] + \Pr\lt[ L_S \!=\! 0 \mid E \rt] = 1$ and thus, one of the two events must occur with probability at least $\tfrac{1}{2}$.

  Now consider two additional ID assignments $I'$ and $I''$ such that $I$, $I'$, and $I''$ are mutually disjoint.
  By a similar argument as before, there exists a set $n/2$ nodes with IDs $S' \subseteq I'$ such that none of these nodes receives or sends any messages if event $E_{I'}$ happens, and we can also identify such a subset $S'' \subseteq I'$, conditioned on $E_{I''}$.
  We define the random variables $L_{S'}$ and $L_{S''}$ in exactly the same way as we did for $S$.

  Since we have 3 ID assignments, it follows that there must be two of them, say $I$ and $I'$, such that $\Pr\lt[ L_{S} \!=\! b \mid E \rt] \ge \tfrac{1}{2}$ and $\Pr\lt[ L_{S'} \!=\! b \mid E \rt] \ge \tfrac{1}{2}$, for some $b \in \set{0,1}$.
  Consider first the case that that $b=1$.

  Now consider the execution of the algorithm on the ID assignment $S \cup S'$.
  Since this execution is indistinguishable for nodes with IDs in $S$ and $S'$ from the executions with IDs $I$ and $I'$, respectively, it follows that the state transitions of every node have the same probability distribution in both executions.
  Consequently, the event that neither the nodes with IDs in $S$ nor the ones with IDs in $S'$ send or receive any messages has nonzero probability to occur.
  Moreover, the events $L_S \!=\! 1$ and $L_{S'} \!=\! 1$ are independent, and thus the algorithm elects $2$ leaders with nonzero probability, which provides a contradiction.

  A similar argument yields a contradiction when $b = 0$, here, this argument shows that the event where there is no leader in the execution $S \cup S'$ has nonzero probability.
\end{proof}
}%

\section{The Synchronous Clique under Adversarial Wake-up} \label{sec:adversarial}
\input{async-lb}

\section{A Communication-Time Tradeoff in the Asynchronous Clique} \label{sec:async}
\input{async}

\subsection{Asynchronizing the Trade-off of Afek and Gafni under Simultaneous Wake-up} \label{sec:async_afek}

\input{asynchronous_afek_and_gafni}

\section{Conclusion}
The problem of leader election has been introduced a long time ago, and has been heavily investigated ever since. Yet, quite a few significant new results have been discovered very recently.
At the time some of the classic results were obtained, it may have seemed that these issues were very well understood.
This suggests that it does makes sense to look again even at classic results such as the those of Afek and Gafni (for tradeoffs), Afek and Matias (for randomized algorithms), and Frederikson and Lynch (for methods of generalizing lower bounds from comparison to general algorithms), which are some of the papers that we revisited in this work (and improved some of their results, even though they may have appeared to be tight at the time).
It will be interesting to know whether some additional classic results can do with a refresh.

An interesting open question raised by our work is the precise impact of the size of the node ID space on the message complexity of deterministic algorithms for leader election and other problems.
As elaborated in Section~\ref{sec:contributions}, showing lower bounds for an ID space of polynomial size is necessary in order to make the lower bound applicable to the standard $\mathsf{CONGEST}$ model.
\onlyLong{
For instance, the work of Awerbuch, Goldreich, Peleg, and Vainish~\cite{awerbuch1990trade} proved that solving single-source broadcast in general networks has a message complexity of $\Omega(m)$, (where $m$ is the number of edges) when assuming an exponentially large ID space, whereas the subsequent work of King, Kutten, and Thorup~\cite{KKT} showed that a message complexity of $\tilde O(n)$ is indeed possible for smaller ID spaces.
Another open problem is whether we can show a tradeoff lower bound in the asynchronous model, similar in spirit to our result in Section~\ref{sec:tradeoff}.
}

\begin{acks}
Shay Kutten was supported in part by grant 2070442 from the ISF.

Xianbin Zhu is a full-time PhD student at the City University of Hong Kong. The work described in this paper was partially supported by a grant from the Research Grants Council of the Hong Kong Special Administrative Region, China [Project No. CityU 11213620].
\end{acks}

\bibliographystyle{ACM-Reference-Format}
\bibliography{references}

\end{document}

%% file: preamble.tex
\usepackage{amsmath,amsthm,amsbsy}
\usepackage{microtype}
\usepackage{setspace}
\usepackage{color}
\usepackage{xcolor}
\usepackage{algorithm}
\usepackage[noend]{algpseudocode}
\usepackage{multirow}
\usepackage{xspace}

\let\originalleft\left
\let\originalright\right
\renewcommand{\left}{\mathopen{}\mathclose\bgroup\originalleft}
\renewcommand{\right}{\aftergroup\egroup\originalright}

\renewcommand{\Pr}{\operatorname*{\textbf{\textup{Pr}}}}

\DeclareMathOperator*{\EE}{\textbf{\textup{E}}}

\renewcommand{\geq}{\geqslant}
\renewcommand{\ge}{\geqslant}
\renewcommand{\leq}{\leqslant}
\renewcommand{\le}{\leqslant}
\newcommand{\set}[1]{\{ #1 \}}

\newcommand{\lb}{\left}
\newcommand{\rb}{\right}
\newcommand{\lt}{\left}
\newcommand{\rt}{\right}

\DeclareMathOperator{\poly}{\ensuremath{\mathrm{poly}}}

\makeatletter
\newtheorem*{rep@theorem}{\rep@title}
\newcommand{\newreptheorem}[2]{%
\newenvironment{rep#1}[1]{%
\def\rep@title{#2 \ref{##1}}%
\begin{rep@theorem}[restated]}%
{\end{rep@theorem}}}
\makeatother
\newreptheorem{lemma}{Lemma}
\newreptheorem{theorem}{Theorem}

\newboolean{short}
\newcommand{\onlyShort}[1]{\ifthenelse{\boolean{short}}{#1}{}}
\newcommand{\onlyLong}[1]{\ifthenelse{\boolean{short}}{}{#1}}

\newtheorem{claim}{Claim}
\theoremstyle{plain}

\usepackage[most]{tcolorbox}
\tcbuselibrary{theorems}
\newtcbtheorem[number within=section]{Definition}{Definition}{
  enhanced,
  sharp corners,
  attach boxed title to top left={
    yshifttext=-1mm
  },
  colback=white,
  colframe=gray!75!white,
  fonttitle=\bfseries,
  boxed title style={
    sharp corners,
    size=small,
    colback=gray!75!white,
    colframe=gray!75!white,
  }
}{def}

\newcommand{\ann}[1]{%
\text{\footnotesize(#1)}\quad}

%% file: intro.tex
We address one of the most fundamental problems in the area of distributed computing---breaking symmetry by electing a leader among the $n$ nodes of a network. This problem was introduced by Le Lann~\cite{le1977distributed} and serves as a crucial building block in numerous applications such as resource allocation, load balancing, etc. The literature on leader election is too rich to survey here even for models close to the one used here, and
the study of leader election has had a significant influence on other areas including e.g., peer-to-peer networks, sensor networks, population protocols, programmable matter, etc.~\cite{nygren2010akamai, ganeriwal2003timing, yao2002cougar,popul,derakhshandeh2015leader}. The main complexity measures studied have been message and time complexities as studied here. Note that for some of the other above mentioned kinds of networks (e.g., sensor networks), saving in energy is also very important. However, saving in messages and time can also help save energy.

In this work, we focus on resolving the complexity bounds of leader election in the synchronous clique.
 However, let us note that clique networks have been studied intensively for their own sake. This includes a wide literature on leader election in clique networks. Some such papers are mentioned below. Clique networks have also been studied for other problems and other models of distributed computing. One reason may be the fact that they capture the network (or the application) layer of most networks, where essentially, every node is able to reach every other node.

\subsection{Our Contributions} \label{sec:contributions}

Our main contributions are new tradeoff and lower bound results that aim at resolving the message and time complexity of leader election in clique networks under various assumptions.
Table~\ref{tab:results} summarizes our main results. 
\begin{table*}[t]
    \centering
\begin{threeparttable}
\begin{small}
  \begin{tabular*}{\textwidth}{l@{\hskip 1cm} l@{\hskip 0.5cm} l c l}
  \toprule
  Result &  Remarks & Time & & Messages\\
  \midrule
  \multicolumn{2}{l}{\textbf{Synchronous, Deterministic, Simultaneous Wake-up}} \\
    Lower Bound, Theorem~\ref{thm:tradeoff} & 
    for any $f(n)>1$ $^{*}$ &
    $\le\frac{\log_2 n - 1}{\log f(n) + 1} + 1$ & $\Rightarrow$ &
    $\ge n\,f(n)$ \\
    Lower Bound, Theorem~\ref{thm:lb_unconditional} & 
    $T(n)$ is any function of $n$ $^{\$}$&
    $\le T(n)$ & $\Rightarrow$ & $\Omega\lt( n\log n \rt)$\\
    Algorithm, Theorem~\ref{thm:improved_afekgafni}  & 
    any odd $\ell\ge3$ &
    $\ell$ & & $O\lt( \ell\,n^{1+\frac{2}{\ell+1}} \rt)$\\
  \midrule
  \multicolumn{2}{l}{\textbf{Synchronous, Deterministic, Adversarial Wake-up}}\\
    Algorithm \cite{afek1991time} & 
    any $\ell\ge2$ &
    $\ell$ & & $O\lt( \ell\,n^{1+\frac{2}{\ell}} \rt)$\\
    Lower Bound \cite{afek1991time} & 
    for any $c\ge 2$ &
    $\le \frac{1}{2}\log_cn$ & $\Rightarrow$ & $\ge\frac{c-1}{2}n\log_cn$\\
    Lower Bound \cite{afek1991time}& 
    $T(n)$ is any function of $n$ &
    $\le T(n)$ & $\Rightarrow$ & $\Omega\lt( n\log n \rt)$\\
  \midrule
  \multicolumn{2}{l}{\textbf{Synchronous, Randomized, Simultaneous Wake-up}}\\
    Algorithm, Theorem~\ref{thm:las_vegas}  & 
    Las Vegas  &
    $3$ (w.h.p.) & & $O(n)$ (w.h.p.)\\
    Lower Bound, Theorem~\ref{thm:las_vegas} & 
    Las Vegas & & &$\Omega\lt( n \rt)$\\
    Algorithm~\cite{kutten2015sublinear}  & 
    succeeds w.h.p.  &
    $2$ & & $O\lt(\sqrt{n}\log^{3/2}(n)\rt)$ \\
    Lower Bound~\cite{kutten2015sublinear} & 
    for small constant error prob.\ 
    & & &$\Omega\lt( \sqrt{n} \rt)$\\
  \midrule
  \multicolumn{2}{l}{\textbf{Synchronous, Randomized, Adversarial Wake-up}}\\
    Algorithm, Theorem~\ref{thm:adversarial_two_round_algo}  & 
    succeeds w.p.\ $1-\epsilon$  &
    $2$ & & $O\lt(n^{3/2}\log\lt( \frac{1}{\epsilon} \rt)\rt)$ $^{\dagger,\dagger\dagger}$\\
    Lower Bound, Theorem~\ref{thm:adversarial_lb} & 
    for any constant error prob.\ & 
    $\le2$ &
    $\Rightarrow$ &
    $\Omega\lt( n^{3/2} \rt)$ $^\dagger$\\
    Algorithm \cite{kutten2020singularly} & 
    succeeds w.h.p. &
    9 & & $O\lt(n\rt)$\\
  \midrule
  \multicolumn{2}{l}{\textbf{Asynchronous, Randomized}} \\
    Algorithm, Theorem~\ref{thm:async} & 
    succeeds w.h.p.; $k\in[2,O\lt( \frac{\log n}{\log\log n} \rt)]$ &
    $k+8$ & & $O\lt(n^{1+1/k}\rt)$\\
    Algorithm \cite{kutten2020singularly} & 
    succeeds w.h.p. &
    $O(\log^2n)$ & & $O\lt(n\rt)$\\
  \bottomrule
\end{tabular*}
\end{small}
\begin{normalsize}
  \begin{tablenotes}
  \item[$*$] Assuming an ID space of size at least $\Omega(n\log n)$.
  \item[\$] Assuming an ID space of size at least $\Omega\lt( n^{\log_2 n}(T(n))^{\log_2 n -1} \rt)$.
  \item[$\dagger$] In expectation.
  \item[$\dagger\dagger$] With high probability, the message complexity does not exceed $O(n^{3/2}\log n)$.
  \end{tablenotes}
\end{normalsize}
\vspace{5 mm}
\caption{Results for Leader Election in Clique Networks. Note that the ``$\Rightarrow$'' symbol for the entries of Theorems~\ref{thm:tradeoff}, \ref{thm:lb_unconditional}, and \ref{thm:adversarial_lb} indicates that a sufficiently small time complexity implies the stated bound on the message complexity.} \label{tab:results}
\end{threeparttable}
\end{table*}
  
\subsubsection{\textbf{Results for the Synchronous Clique under Simultaneous Wake-up}}
\paragraph{\textsl{\textbf{Improved Tradeoff between Messages and Time}}}
In Section~\ref{sec:tradeoff}, we focus on the setting where every node is awake at the start of round $1$.
We show that there is an inherent tradeoff between messages and time by proving that any deterministic leader election algorithm takes at least $\frac{\log_2 n - 1}{\log_2 f(n) + 1}$ rounds, if it sends at most $n\cdot f(n)$ messages, for $f(n)>1$, and therefore any $\ell$-round algorithm requires at least $\Omega \lt( n^{1 + {1}/(\ell-1)} \rt)$ messages.
We point out that, in contrast to previous lower bounds for general\footnote{that is, not necessarily comparison-based algorithm}
deterministic algorithms under simultaneous wake-up, such as the celebrated message complexity lower bound of $\Omega(n\log n)$ for leader election in rings by Frederickson and Lynch~\cite{frederickson1987electing}, our result does not use Ramsey's theorem, and thus holds even for moderately large ID spaces, i.e., of size $\Theta(n\log n)$. This difference is meaningful in the $\mathsf{CONGEST}$ model~\cite{peleg2000distributed}, where assuming a large ID space implies a large communication complexity trivially.

In Section~\ref{sec:improved_afekgafni}, we give a nearly matching upper bound via a simple modification of the deterministic algorithm by Afek and Gafni~\cite{afek1991time}, obtaining an $\ell$-round algorithm, for any odd $\ell = 2k -3 \ge 3$ that sends $O \lt( \ell\,n^{1 + 2/(\ell+1)} \rt)$ messages.

\paragraph{\textsl{\textbf{An $\Omega(n\log n)$ Lower Bound for Time-Bounded Algorithms}}} In Section~\ref{sec:unconditional}, we prove
a lower bound of $\Omega(n\log n)$ messages that holds for any time-bounded algorithm, i.e., for algorithms with a time complexity that is a function of $n$. We also show that the assumption of a sufficiently large ID set
 (though significantly smaller than the one used in ~\cite{frederickson1987electing}) is indeed crucial for a message complexity lower bound of $\Omega(n\log n)$ to hold, by giving a deterministic algorithm that terminates in sublinear time and sends $o(n\log n)$ messages when the ID space has linear size.

\paragraph{\textsl{\textbf{An $\Omega(n)$ Lower Bound for Las Vegas Algorithms}}}
In Section~\ref{sec:las_vegas}, we show that any randomized Las Vegas leader election algorithm has a message complexity of $\Omega(n)$ and argue that this bound is tight.
This reveals a polynomial gap in the message complexity compared to the Monte Carlo algorithm of Kutten et al.~\cite{kutten2015sublinear} which succeeds with high probability while terminating in only $2$ rounds and sending $O(\sqrt{n}\log ^{3/2}n)$ messages.

\subsubsection{\textbf{A Tight Bound for the Synchronous Clique under Adversarial Wake-up}}
In Section~\ref{sec:adversarial}, we turn our attention to the setting where the adversary wakes up an arbitrary set of nodes, while all other nodes are asleep initially, and awake only upon receiving a message.
We show that any randomized $2$-round algorithm that succeeds with constant probability must send $\Omega(n^{3/2})$ messages in expectation, and give a simple algorithm that tightly matches this bound.
To the best of our knowledge, this is the first lower bound that holds for \emph{randomized} leader election algorithms under adversarial wake-up.

\subsubsection{\textbf{Results for the Asynchronous Clique}}

In Section~\ref{sec:async}, we give the first algorithm that achieves a tradeoff between messages and time in the asynchronous model, again, under the assumption that some subset of nodes is awoken by the adversary.
In particular, we show that, for any parameter choice of $k \in [2,O(\log n / \log\log n)]$, there is an algorithm that terminates in $k+8$ rounds while sending $O(n^{1 + 1/k})$ messages.

The above tradeoff algorithm is randomized. As a second asynchronous tradeoff, in Section~\ref{sec:async_afek} we translated the original deterministic tradeoff algorithm (and tradeoff upper bounds) of Afek and Gafni~\cite{afek1991time} to the limited asynchronous model. That is, for the sake of this (deterministic) algorithm only, we assume that indeed, the delay of each message is still chosen by the adversary, however, the time complexity for this algorithm is counted from the last spontaneous wake up of a node (or, alternatively, assuming all the nodes wake up simultaneously). This partially answers an old open problem posed by Afek and Gafni.

\subsection{Related Work} \label{sec:related}

A result that is closely related to our tradeoff lower bound is the seminal work of Afek and Gafni~\cite{afek1991time}, who showed (see Theorem 4.5 in \cite{afek1991time}) that if a deterministic algorithm elects a leader in at most $\frac{1}{2}\log_c(n)$ rounds of the synchronous clique under adversarial wake-up, then it must send at least $\frac{c-1}{2}n\log_c n$ messages, for any $c\ge 2$.
However, their lower bound uses an adversarially constructed execution that crucially relies on the fact that nodes do \emph{not} wake up simultaneously.
Hence, their techniques do not apply to the setting that we consider in Section~\ref{sec:tradeoff}, where all the nodes are assumed to start the execution \emph{simultaneously}. %

 For the case of adversarial wake up, where the lower bounds of~\cite{afek1991time} apply too, their tradeoff is not directly comparable to the one obtained in
in our Theorem~\ref{thm:tradeoff}.
In more detail, consider any $k$-round deterministic algorithm, where $2 \le k = O(\log n)$.
(Note that any $1$-round algorithm must trivially send $\Theta(n^2)$ messages.)
Our Theorem~\ref{thm:tradeoff} shows that any $k$-round deterministic algorithm must send $\Omega( (n/2)^{1 + \frac{1}{k-1}})$ messages, whereas the message complexity lower bound of \cite{afek1991time} is $\Omega \lt( k\,n^{1 + {1}/{2k}} \rt)$.
Thus, for constant time algorithms, our result improves by a polynomial factor over the previously best lower bound.
On the other hand, for $k = \Theta(\log n)$, the result of \cite{afek1991time} yields a bound that is greater by a factor of $\Theta(\log n)$. %

Afek and Gafni~\cite{afek1991time} also provide an $\Omega(n\log n )$ lower bound on the message complexity that was not conditioned by the bound on the time. However, again,
 their proof relies heavily on the assumption that the adversary can choose which nodes to wake up and when, which is not admissible in our setting; the $\Omega(n\log n)$ lower bound we provide in Theorem~\ref{thm:lb_unconditional} holds even when assuming simultaneous wake up. An earlier $\Omega(n \log n)$ lower bound for \emph{asynchronous} networks was shown by Korach, Moran, and Zaks \cite{korach1984tight}. Again, the wake-up was adversarial (since the network was asynchronous). Moreover, a communication lower bound for asynchronous networks does not translate easily to one in synchronous networks, since in synchronous networks, it is possible to convey information using the passage of time.

Afek and Gafni also pose the open problem whether a sublinear time message optimal asynchronous algorithm exists. In particular, when they translate their synchronous tradeoff algorithm into the asynchronous model, they are forced to increase its time complexity to linear.
The also explain and demonstrate that ``The arbitrary delay of messages
(and not the absence of the clock) is the source of the increase in the time complexity of
the algorithm.''
We managed to translate their algorithm to the asynchronous model, preserving the very same tradeoff they obtained between time and number of messages.
The translated algorithm does cope well with the arbitrary delay of messages, thus overcoming the hurdle they posed.
 However, for this algorithm only, we assume that all the nodes wake up simultaneously. Alternatively, the translated algorithm exhibit that tradeoff (in the asynchronous model) if one counts the time only starting from the time of the last spontaneous wake up.

 Previous work on leader election also considered the trade-off between message complexity and time complexity on networks with a regular structure. Frederickson \cite{frederickson1983tradeoffs} gave trade-offs on special network topologies, including rings, meshes and trees.

For clique networks, it was shown by Korach, Moran, and Zaks \cite{korach1984tight} that the existential lower bound of $\Omega(m)$
(eventually established in \cite{kutten2015complexity}) does not hold, since they presented an algorithm with $O(n \log n)$ message complexity, even though in a clique, $m=\Theta(n^2)$. Such deterministic algorithms were also presented by others, e.g., Humblet
\cite{humblet1984selecting}. One explanation was a graph property that complete networks have (linear traversability) as demonstrated in \cite{korach1990modular}.

Afek and Matias~\cite{afek1994elections} presented a randomized algorithm for the \emph{asynchronous} model, but assuming that all the nodes wake up at the same time.
(Their algorithm requires each node to take actions initially, to create a sparse graph.) On the average, the message complexity is $O(n)$ and the time complexity is $O(\log n)$.
Ramanathan et al.~\cite{ramanathan2007randomized} proposed a randomized \emph{synchronous} leader election algorithm with error probability $O(1/\log ^{\Omega(1)}n)$ that has $O(\log n)$ time complexity and a linear message complexity.
\onlyLong{Subsequently, Kutten et al.~\cite{kutten2015sublinear} gave a randomized \emph{synchronous} algorithm that terminates in just $2$ rounds, while using $O(\sqrt{n}\log ^{3/2}n)$ messages with high probability, and elects a unique leader with high probability.}\xspace
Recently, Kutten et al.~\cite{kutten2020singularly} presented a randomized \emph{asynchronous} algorithm that uses $O(n)$ messages and $O(\log^2 n)$ time under adversarial wake up.
For the same model, Singh~\cite{singh1992leader} presented a lower bound for deterministic algorithms, which shows that achieving $O(n)$ messages comes at the price of having a time complexity of $\Omega(n / \log n)$.

%% file: preliminaries.tex
\section{Preliminaries}

To solve leader election, exactly one of the nodes in the network needs to be distinguished as the leader.
In the \emph{explicit} variant of leader election, every node must output the ID of the leader when it terminates, whereas in \emph{implicit leader election}
(known to be in some cases easier \cite{kutten2015sublinear}),
each node must irrevocably output a single bit indicating whether it is the leader, and the goal is that exactly one of the bits is set to $1$.
Since any algorithm that achieves explicit leader election also solves the implicit version, we focus on implicit leader election whenever  showing a lower bound.

We consider a clique network of $n$ nodes that communicate via synchronous message passing using point-to-point links.
The computation proceeds in rounds and in each round $r \ge 1$, a node can send (possibly distinct) messages to every other node.
When we are mainly concerned with
showing lower bounds, we consider the $\mathsf{LOCAL}$ model~\cite{peleg2000distributed} which means that we do not restrict the size of these messages.
However, our algorithms have their claimed complexities also under the CONGEST model.
Each node has a unique ID and we assume the most common model which is that of the
\emph{clean network model}~\cite{peleg2000distributed} (sometimes called KT$_0$) that restricts the knowledge of each node as follows: initially, a node is aware only of its own ID, and the total number of nodes $n$, but it does not know the IDs of the other nodes.
The nodes are connected via bidirectional links, each connecting two nodes. To address the links, each node has a set of $n-1$ ports, over which it receives and sends messages, respectively.
The assignment of port numbers to destinations is arbitrary and may differ for different nodes.
Formally, a \emph{port mapping} $p$ is a function that maps each pair $(u,i)$ to some pair $(v,j)$, i.e., $p((u,i)) = (v,j)$ means that a message sent by node $u$ over port $i$ is received by node $v$ over port $j$. We note that port mapping is bijective, that is, if $p((u,i)=(v,j)$ then $p((v,j))=(u,i)$.
Note that neither $u$ nor $v$ are aware of how their ports are connected until they send/receive the first message over these ports.
We say that a port $i$ is \emph{unused} if no message was sent or received over $i$
and we say that a node $u$ \emph{opens a port $i$ in round $r$} if $u$ sends the first message over $i$ in $r$.
When analyzing the execution of the algorithm for a limited subset of nodes $S$ and a restricted number rounds up to some number $r$, it may not be necessary for us to fully specify all port mappings, but rather only the mappings of the ports over which actual messages were sent/received by nodes in $S$ until round $r$, while leaving the remaining assignments undefined.
We call such a partially defined function a \emph{partial port mapping}.
If a port mapping $q$ extends a given partial port mapping $p$ by assigning some previously undefined ports, we say that $q$ is \emph{compatible} with $p$.

%% file: async-lb.tex
We now consider the synchronous clique under \emph{adversarial wake-up}, which means that the adversary chooses a nonempty subset of nodes that start the execution in round $1$. At any later round, the adversary may choose additional nodes that are not awake (if such exist) and wake them up too. Any node not woken up by the adversary is asleep initially, and only wakes up at the end of a round if it received a message in that round.
To make the proofs easier to read, we actually make the additional unnecessary assumption that the adversary only wakes up nodes in round 1. For the specific algorithm we present in this section, it is easy to get rid of this assumption. For the lower bound, since it holds under the unnecessary assumption, it certainly holds without it.

\onlyLong{We first give a simple randomized $2$-round leader election algorithm that succeeds with probability $1 - \epsilon - \tfrac{1}{n}$, has $O \lt( n^{3/2}\log \lt(\tfrac{1}{\epsilon}\rt) \rt)$ message complexity in expectation and never sends more than $O\lt(n^{3/2}\log n\rt)$ messages with high probability.
Then, we show that this bound is tight, by proving that $\Omega\lt( n^{3/2} \rt)$ is a lower bound on the message complexity for randomized $2$-round algorithms that succeed with sufficiently large probability.
In fact, this lower bound holds not only for leader election, but even just for the problem of waking-up every node in the network.
}%

\subsection{An Optimal $2$-Round Algorithm with $O(n^{3/2}\log n)$ Message Complexity} \label{sec:adversarial_two_round_algo}
\label{sub:async2round}
Fix some parameter $\epsilon \ge \frac{1}{\poly(n)}$; the resulting algorithm will succeed with probability at least $1 - \epsilon - \tfrac{1}{n}$.
A node that is awoken by the adversary sends $\lceil\sqrt{n}\rceil$ messages over uniformly at random sampled ports (without replacement).
In round $2$, any node that was awoken by the receipt of a round-$1$ message (i.e., \emph{not} by the adversary) becomes a candidate with probability $\frac{\log(1/\epsilon)}{\lceil\sqrt{n}\rceil},$ whereas non-candidate nodes immediately become non-leaders.
Every candidate samples a random rank from $[n^4]$ and sends its rank to all other nodes in round $2$.
At the end of round $2$, a candidate becomes leader if and only if all rank messages that it has received carried a lower rank than its own.
\onlyShort{
In the full paper~\cite{DBLP:journals/corr/abs-2301-08235}, we prove the following theorem:
}

\begin{theorem} \label{thm:adversarial_two_round_algo}
  Consider the synchronous clique under adversarial wake-up and fix any small $\epsilon \ge \frac{1}{\poly(n)}$.
  There is a $2$-round randomized leader election algorithm that succeeds with probability at least $1 - \epsilon - \tfrac{1}{n}$ and has an expected message complexity of $O \lt( n^{3/2}\log\lt(\frac{1}{\epsilon}\rt) \rt)$.
  Moreover, it sends at most $O(n^{3/2}\log n)$ messages (whp).
\end{theorem}
\onlyLong{
\begin{proof}
  Since the adversary wakes up at least one node at the start of round $1$, at least $\lceil \sqrt{n} \rceil$ nodes will be awoken by a message and hence try to become candidate at the start of round $2$.
  The probability that none of the nodes becomes a candidate is at most
  \[
    \lt( 1 - \frac{\log(1/\epsilon)}{\lceil \sqrt{n} \rceil} \rt)^{\lceil \sqrt{n}\rceil}
    \le
    \exp \lt( - \log \lt( \frac{1}{\epsilon} \rt) \rt)
    = \epsilon.
  \]
  Consider some candidate $u$. The probability that all of the at most $n-1$ other candidates sample a rank different than $u$'s is at least $\lt(1 - \frac{1}{n^4}\rt)^n \ge e^{(-2/n^3)} \ge 1 - \tfrac{2}{n^3}$.
  Taking a union bound implies that all candidates have distinct ranks with probability at least $1 - \tfrac{2}{n^2} \ge 1 - \tfrac{1}{n}$.
  Thus, the probability that there is at least one candidate and that all candidates have unique ranks is at least $1 - \epsilon - \tfrac{1}{n}$.
  Conditioned on these two events, it is clear that each candidate learns the ranks of all candidates by the end of round $2$, and hence there is exactly one leader.

  For the message complexity bound, we observe that, in the worst case, every node is awake in round $1$ and sends $O(\sqrt{n})$ messages.
  In round $2$, the expected number of candidates is $\sqrt{n}\log \lt( \frac{1}{\epsilon} \rt)$, and
  each of these candidates sends $n-1$ messages, which shows that the expected message complexity is $O\lt(n^{3/2} + n^{3/2}\log \lt( \frac{1}{\epsilon} \rt)\rt)$.

  Since $\epsilon \ge \frac{1}{\poly(n)}$ and hence $\log \lt( \frac{1}{\epsilon} \rt) = O(\log n)$, a standard Chernoff bound shows that the number of candidates is $O(\log n)$, and hence we obtain the claimed high-probability bound of $O(n^{3/2}\log n)$ on the message complexity.
\end{proof}
}

\subsection{A Tight Lower Bound for Randomized Algorithms} \label{sec:adversarial_lb}

We now prove a lower bound that holds not only for leader election, but even just for the problem of waking-up every node in the network.
In the proof, we assume that nodes are anonymous and do not have IDs.
However, recall that it is straightforward for a randomized algorithm to generate IDs that are unique with high probability, by instructing each node to independently sample a random ID from the range $[n^4]$ (see the proof of Theorem~\ref{thm:adversarial_two_round_algo} for a detailed argument).

\begin{theorem} \label{thm:adversarial_lb}
  Consider the synchronous clique, where an arbitrary nonempty subset of nodes is awoken by the adversary at the start of round $1$.
  Any randomized $2$-round algorithm that wakes up every node with  probability at least $1 - \epsilon$, has an expected message complexity of $\Omega(n^{3/2})$, for any positive constant $\epsilon < 1$.
\end{theorem}
We point out that Theorem~\ref{thm:adversarial_lb} extends to algorithms in the \emph{asynchronous} clique that terminate within $2$ units of time.
Moreover, it improves significantly even over the previously-best lower bound of $\Omega(n^{5/4})$ messages~\cite{afek1991time} for \emph{deterministic} $2$-round algorithms under adversarial wake-up; see Section~\ref{sec:related} for more details about their result.

\subsubsection{Proof of Theorem~\ref{thm:adversarial_lb}}
Since the proof is somewhat involved, let us first explain its highlights informally.
The proof exploits the fact that limiting the time to $2$ rounds, makes it difficult for a node to learn how many other nodes are awake, and hence it may need to send as many messages (on average) as it would in an execution where it is the only one who was woken up.
Intuitively, this is why we are able to demonstrate a gap between the message complexity for a small constant time and the case of a large constant time. That is,
for say, 9 rounds algorithms, $O(n)$ messages are possible \cite{kutten2020singularly}. This gap stands in contrast to the results shown by Afek and Gafni~\cite{afek1991time} for deterministic algorithms. There, for \emph{any constant} time complexity of $t$ rounds, the messages complexity was $\Omega(n^{1+t/2})$ messages. (Recall also that we have shown in Subsection \ref{sub:async2round} above an optimal $O(n^{3/2})$ messages randomized algorithm even for the 2-rounds case.)
The source of the high complexity in the randomized case of 2 rounds algorithms seems different than the one exploited in the lower bound proof of \cite{afek1991time}.
(Moreover, formalizing the argument for randomized algorithms that may fail with some small probability turns out to be more involved.)
Our high-level strategy is as follows:
We call the nodes that are awake initially \emph{roots} and any node that receives a message from a root becomes a \emph{child} of that root (and also of any other root that sent it a message).
We first show that root nodes cannot afford to send more than $o(\sqrt{n})$ messages, and then use this fact to show that many root nodes and child nodes will operate ``undisturbed'', in the sense that they do not receive any additional message on top of what they would have received in an execution where only a single root was awake.
We show that, in expectation, this requires the child nodes to send at least $\Omega(n^{3/2})$ messages to ensure every other node wakes up by the end of round $2$.

We now present the formal argument.
Assume towards a contradiction that there exists a $2$-round algorithm with a message complexity of $\frac{n^{3/2}}{f(n)}$, for some function $f(n) = \omega(1)$.

\paragraph{Events and Notation:}
Our proof relies on the occurrence of certain probabilistic events for which we introduce some notation.
\begin{itemize}[leftmargin=1em]
  \item For any node $u$, let random variable $M(u,r)$ denote the set of round $r$ messages sent by $u$ annotated by the respective port numbers, and let $m(u,r) = |M(u,r)|$.
  We define $m(r) = \sum_{u} m(u,r)$ to be the total number of ports over which a message is sent in round $r$.

  \item Suppose that the adversary wakes up a set of nodes $R=\set{u_1,\dots,u_\ell}$.
  We say that each $u_i$ is a \emph{root}, and define $u_i$'s \emph{children} to be the set of nodes $\{v_{i,1},\dots,v_{i,M(u_i,1)}\}$ that were asleep until receiving a message from $u_i$ in round $1$.
  Note that it is possible that $v_{i,j}$ is the child of multiple roots.

  \item $\Few_{u}$ captures the event that node $u$ sends few messages in round $1$. Formally, $\Few_{u}$ happens iff node $u$ is a root and sends at most $\frac{\sqrt{n}}{g(n)}$ messages, for a function $g(n) = \omega(1)$, where $g(n) \ge f(n)$.
  We also define $\Few_{R} = \bigwedge_{u \in R}\Few_u$, for the set of nodes $R$.

  \item Event $\Undist_{v}$ (``$v$ is \emph{undisturbed}'') occurs iff one of the following cases is true:
    \begin{enumerate}
      \item $v$ is a root and does not receive any messages in round $1$.
      \item $v$ is a child of some root $u_i$ and only receives a message in round $1$ from $u_i$. %
    \end{enumerate}
\end{itemize}

Let us start the proof of Theorem~\ref{thm:adversarial_lb} by showing several bounds on the probability of the events introduced above. \onlyShort{(See the full paper~\cite{DBLP:journals/corr/abs-2301-08235} for omitted proofs.)}\xspace Throughout this section, we define $\Lambda$ to be an execution in which the adversary awakens a set of $\lceil \sqrt{n} \rceil$ root nodes $R$, and $R' \subseteq R$  the subset of roots that remain undisturbed in round $1$ of $\Lambda$, i.e., $\bigwedge_{u_i \in R'} \Undist_{u_i}$.

\begin{lemma} \label{lem:all_few}
  Let $R$ be the set of roots and suppose that $|R| = \Theta(\sqrt{n})$.
  There exists a function $g(n) = \omega(1)$ such that
  $\Pr\lt[ \Few_R\  \rt] \ge 1 - O\lt( \tfrac{1}{g(n)} \rt)$.
\end{lemma}
\onlyLong{
\begin{proof}
  Let $\Gamma$ be the execution where every node is a root, i.e., awoken in round $1$.
  Since $M(u_i,1)$ only depends on $u_i$'s private coin flip and, in particular, is independent of the round-$1$ computation of other nodes, it
  follows that $\Pr\lt[ \Few_{u_i}\rt] = \Pr\lt[ \Few_{u_j}\rt]$ for all nodes $u_i$ and $u_j$, and for both executions $\Lambda$ and $\Gamma$.
  Assume towards a contradiction that $\Pr\lt[ \neg\Few_{u_i} \rt] \ge \frac{c}{\sqrt{n}}$, for some constant $c>0$.
  Since $n$ nodes are awake, it follows that the expected value of $m(1)$ is $\Omega \lt( n^{3/2} \rt)$ in both executions $\Gamma$ and $\Lambda$, which is a contradiction.

  So far, we have shown that $\Pr\lt[ \neg\Few_R \rt] \le \frac{1}{\sqrt{n}\,g(n)}$, for some function $g(n) = \omega(1)$ and, as argued above, this also holds in the execution $\Lambda$ where only the nodes in $R$ are roots.
  It follows that, in execution $\Lambda$,
  \begin{align*}
    \Pr\lt[ \Few_R \rt]
    = 1 - \Pr\lt[ \exists u \in R\colon \neg\Few_u \rt]
    \ge 1 - \tfrac{|R|}{\sqrt{n}\,g(n)}
    \ge 1 - O\lt( \tfrac{1}{g(n)} \rt).
  \end{align*}
\end{proof}
}

\begin{lemma} \label{lem:undisturbed_roots}
  Let $R'\subseteq R$ be the set of undisturbed roots, i.e., it holds that $\bigwedge_{u_i \in R'} \Undist_{u_i}$.
  Then,
  \begin{enumerate}
    \item[(A)] $\Pr\lt[ \Undist_{u_i}\ \middle|\ \Few_R \rt]
      \ge 1 - O \lt( \tfrac{1}{g(n)} \rt),$ and
    \item[(B)] $\Pr\lt[ |R'| \ge \tfrac{1}{2}|R|\ \middle|\ \Few_R \rt] \ge 1 - O \lt( \tfrac{1}{g(n)} \rt).$
  \end{enumerate}
\end{lemma}
\onlyLong{
\begin{proof}
  Since we condition on event $\Few_R$, we know that every  root sends at most $\sqrt{n}/g(n)$ messages in round $1$, and thus the total number of messages sent in round $1$ is at most $\frac{|R|\sqrt{n}}{g(n)}$.
  Consider any root $u_i \in R$.
  Since the port connections are chosen uniformly at random, the probability that any message of another root $u_j$ reaches $u_i$ is at most $\frac{1}{\sqrt{n}\,g(n)}.$
  Thus,
  \begin{align*}
    \Pr\lt[ \Undist_{u_i} \ \middle|\ \Few_R\rt]
    \ge
    \lt( 1 - \frac{1}{\sqrt{n}\,g(n)} \rt)^{|R|}
    \ge
    \exp \lt( - \frac{2|R|}{\sqrt{n}\,g(n)} \rt)
    \ge
    1 - O \lt( \frac{1}{g(n)} \rt),
  \end{align*}
  which proves (A).
  For (B), note that
  $\EE\lt[ |R| - |R'| \ \middle|\ \Few_R \rt]
  \le
    O \lt( \frac{|R|}{g(n)} \rt),$
  and, by applying Markov's Inequality, we derive that
  $\Pr\lt[ |R| - |R'| \ge \tfrac{1}{2}|R| \rt]
    \le O \lt( \frac{1}{g(n)} \rt).
  $
\end{proof}
}

We are now ready to derive a lower bound on the message complexity of the assumed algorithm.
To this end, we lower-bound the total expected message complexity by only counting the number of messages sent by the children of each undisturbed root $u_i \in R'$ in round $2$.
We have
\begin{align}
  \EE\lt[ m(2) \rt]
  &\ge
    \EE\lt[
      \sum\onlyShort{\nolimits}_{u_i \in R'}
      \sum\onlyShort{\nolimits}_{j=1}^{m(u_i,1)}
      m(v_{i,j},2)
    \rt]\ \notag\\
  &\ge
    \Pr\lt[ \Few_R \rt]
    \EE\lt[
      \sum\onlyShort{\nolimits}_{u_i \in R'}
      \sum\onlyShort{\nolimits}_{j=1}^{m(u_i,1)}
      m(v_{i,j},2)
    \ \middle|\
      \Few_R
    \rt].\notag\\
  \ann{by Lemma~\ref{lem:all_few}}
  &\ge
    \lt(1 - o(1)\rt)
    \EE\lt[
      \sum\onlyShort{\nolimits}_{u_i \in R'}
      \sum\onlyShort{\nolimits}_{j=1}^{m(u_i,1)}
      m(v_{i,j},2)
    \ \middle|\
      \Few_R
    \rt].\notag
\end{align}
Next, we condition on the event that $|R'| \ge \tfrac{1}{2}|R|$ and, without loss of generality, assume that $R' = \set{u_1,\dots,u_{\ell}}$, where $\ell = \lfloor |R|/2 \rfloor$.
We get
\begin{align}
  \EE\lt[ m(2) \rt]
  &\ge
    (1 - o(1))
    \Pr\lt[ |R'| \!\ge\! \tfrac{1}{2}|R| \ \middle|\ \Few_R  \rt]\notag\\
    &\cdot
    \sum_{i=1}^\ell
    \EE\lt[
      \sum_{j=1}^{m(u_i,1)}
      m(v_{i,j},2)
    \ \middle|\
      |R'| \!\ge\! \tfrac{1}{2}|R|,
      \Few_R,
      \Undist_{u_i}
    \rt]. 
    \notag\\
  \ann{by Lem.~\ref{lem:undisturbed_roots}}
  &\ge
  (1 - o(1))\notag\\ 
  &\cdot\sum_{i=1}^\ell
  \EE\lt[
    \sum_{j=1}^{m(u_i,1)}
    m(v_{i,j},2)
  \ \middle|\
    |R'| \!\ge\! \tfrac{\sqrt{n}}{2},
    \Few_R,
    \Undist_{u_i}
  \rt]. \label{eq:exp_bnd6}
\end{align}
Note that the conditioning on $\Undist_{u_i}$ is implied by the fact that the outer sum ranges only over undisturbed roots.
Let $C_i'$ be the set of undisturbed child nodes of root $u_i$.
We make use of the following lemma\onlyLong{, whose proof we postpone to Section~\ref{sec:good}:}\onlyShort{:}

\newcommand{\lemmaGood}{%
Consider execution $\Lambda$.
For every root $u_i \in R$, we have
\[
\Pr\lt[  |C_i'| \!\ge\! \tfrac{1}{2}m(u_i,1),  m(u_i,2) \!=\! o(n)  \ \middle|\ |R'|\!\ge\! \tfrac{\sqrt{n}}{2}, \Few_R,
  \Undist_{u_i}
 \rt]
\ge 1 - o(1).
\]
}
\begin{lemma} \label{lem:good}
  \lemmaGood
\end{lemma}

Let $\mathcal{E}_i$ denote the event
\onlyLong{\begin{align}%
  \lt( |C_i'| \!\ge\! \tfrac{1}{2}m(u_i,1) \rt) \land  \lt(  m(u_i,2) \!=\! o(n) \rt) \land (|R'|\!\ge\! \tfrac{\sqrt{n}}{2} ) \land  \Few_R
  \land  \Undist_{u_i}
  \label{eq:cond}
\end{align}}%
\onlyShort{
$
  \lt( |C_i'| \!\ge\! \tfrac{1}{2}m(u_i,1) \rt) \land  \lt(  m(u_i,2) \!=\! o(n) \rt) \land (|R'|\!\ge\! \tfrac{\sqrt{n}}{2} ) \land  \Few_R
  \land  \Undist_{u_i}
$
}\xspace
where we assume (w.l.o.g.) that $C_i' = \set{v_{i,1},\dots,v_{i,k}}$, for some $k \ge \tfrac{1}{2}m(u_i,1)$.
By restricting the innermost summation to range only over the undisturbed children in \eqref{eq:exp_bnd6}, we get
\begin{align}
  \EE\lt[ m(2) \rt]
  &\ge
    (1 - o(1))\notag\\ 
    \ &\cdot\Pr\lt[  |C_i'| \!\ge\! \tfrac{m(u_i,1)}{2},  m(u_i,2) \!=\! o(n)\ \middle|\  |R'|\!\ge\! \tfrac{\sqrt{n}}{2}, \Few_R\rt]\notag\\ 
    \quad &\cdot\sum_{i=1}^\ell
    \EE\lt[
      \sum_{j=1}^{k}
      m(v_{i,j},2)
    \quad \middle|\
      \mathcal{E}_i
    \rt]
  \notag\\
  \ann{by Lem~\ref{lem:good}}
  &\ge
  \begin{multlined}[t]
    (1 - o(1))
    \sum\onlyShort{\nolimits}_{i=1}^\ell
    \EE\lt[
      \sum\onlyShort{\nolimits}_{j=1}^{k}
      m(v_{i,j},2)
    \ \middle|\
      \mathcal{E}_i
    \rt].
  \end{multlined} \label{eq:exp_bnd10}
\end{align}

Let $\Corr$ be the event that the algorithm succeeds in waking up all nodes.
Notice that
\begin{align}
  \Pr\lt[ \Corr\ \middle|\ \mathcal{E}_i \rt]
  &\ge \Pr\lt[ \Corr \rt]
  -
  \Pr\lt[ \neg \mathcal{E}_i \rt] \notag\\ 
  &\ge
    \lt( 1 - \epsilon \rt)
    - o(1).\label{eq:corr}
\end{align}

Continuing from the right-hand side of \eqref{eq:exp_bnd10}, we get
\begin{align}
  \EE\lt[ m(2) \rt]
  &\ge
    (1 - o(1)) \notag\\ 
   &\quad\cdot\Pr\lt[ \Corr \ \middle|\ \mathcal{E}_i \rt]
    \sum\onlyShort{\nolimits}_{i=1}^\ell
    \EE\lt[
      \sum\onlyShort{\nolimits}_{j=1}^{m(u_i,1)}
      m(v_{i,j},2)
    \ \middle|\
      \Corr, \mathcal{E}_i
    \rt], \notag\\ 
  \ann{by \eqref{eq:corr}}
  &\ge
  \begin{multlined}[t]
    (1 - \epsilon - o(1))
    \sum\onlyShort{\nolimits}_{i=1}^\ell
    \EE\lt[
      \sum\onlyShort{\nolimits}_{j=1}^{m(u_i,1)}
      m(v_{i,j},2)
    \ \middle|\
      \Corr, \mathcal{E}_i
    \rt].
  \end{multlined} \label{eq:exp_bnd8}
\end{align}
Let $\Gamma$ be the execution where only $u_i$ is awake in round $1$.
Due to conditioning on $\mathcal{E}_i$\onlyShort{,}\onlyLong{(defined in \eqref{eq:cond}),}\xspace it holds that
\onlyLong{\begin{align*}
  m(u_i,1) + m(u_i,2) \le \frac{\sqrt{n}}{g(n)} + o(n) = o(n),
\end{align*}}%
\onlyShort{
$
  m(u_i,1) + m(u_i,2) \le \frac{\sqrt{n}}{g(n)} + o(n) = o(n),
$
}\xspace
and hence, conditioned on $\Corr$, we know that
\begin{align}
  \sum\onlyShort{\nolimits}_{j=1}^{k} m(v_{i,j},2) \ge n - o(n) \ge \frac{n}{2} \label{eq:child_msgs}
\end{align}
in $\Gamma$, i.e., the children of $u_i$ become responsible for waking up $\Omega(n)$ nodes in round $2$.
Since $u_i \in R'$, we know that $u_i$ is undisturbed, and thus the distribution of $m(u_i,2)$ is the same in both executions $\Gamma$ and $\Lambda$ (in which there are $\Theta(\sqrt{n})$ roots).
Furthermore, the summation in \eqref{eq:child_msgs} ranges only over the undisturbed child nodes, which means that $m(v_{i,j},2)$ has the same distribution in $\Gamma$ and $\Lambda$.
Therefore, the bound~\eqref{eq:child_msgs} must also hold in $\Lambda$, for the children of each $u_i \in R'$.
Plugging \eqref{eq:child_msgs} into \eqref{eq:exp_bnd8} and recalling that $\ell = \Theta \lt( \sqrt{n} \rt)$, we conclude that
\onlyLong{%
\begin{align*}
  \EE\lt[ m(2) \rt]
  \ge
    (1 - \epsilon - o(1))
    \frac{\ell\,n}{2}
  = \Omega \lt( n^{3/2} \rt).
\end{align*}}%
\onlyShort{
$
  \EE\lt[ m(2) \rt]
  \ge
    (1 - \epsilon - o(1))
    \frac{\ell\,n}{2}
  = \Omega \lt( n^{3/2} \rt).
$
}\xspace

\onlyLong{
\subsubsection{Proof of Lemma~\ref{lem:good}} \label{sec:good}
\begin{replemma}{lem:good}
  \lemmaGood
\end{replemma}
The following Lemmas~\ref{lem:root_msgs} and \ref{lem:undisturbed_children} together with a simple application of the chain rule will complete the proof of Lemma~\ref{lem:good}.

\begin{lemma} \label{lem:root_msgs}
  Let $u_i \in R$ be a root in execution $\Lambda$.
  Then,
  \begin{align*}
    \Pr\lt[ m(u_i,2) = o(n)
    \ \middle|\
      |R'| \ge \tfrac{\sqrt{n}}{2}, \Few_R, \Undist_{u_i}
    \rt] \ge 1 - o(1).
  \end{align*}
\end{lemma}
\begin{proof}

  Let event $\mathcal{F}_i = (|R'| \ge \tfrac{\sqrt{n}}{2}) \land \Few_R \land  \Undist_{u_i}$.
  Note that
  \begin{align*}
    \Pr\lt[ (|R'| \!\ge\! \tfrac{\sqrt{n}}{2}) \land \Few_R \land  \Undist_{u_i} \rt]
    &=
    \Pr\lt[
        |R'| \!\ge\! \tfrac{\sqrt{n}}{2}
      \ \middle|\
        \Undist_{u_i},
        \Few_R
    \rt]
    \Pr\lt[
      \Undist_{u_i}
    \ \middle|\
      \Few_R
    \rt]
    \Pr\lt[ \Few_R \rt]\notag\\
    &\ge
    \Pr\lt[
        |R'| \!\ge\! \tfrac{\sqrt{n}}{2}
      \ \middle|\
        \Few_R
    \rt]
    \Pr\lt[
      \Undist_{u_i}
    \ \middle|\
      \Few_R
    \rt]
    \Pr\lt[ \Few_R \rt],\notag
  \end{align*}
  where the last inequality follows because removing the conditioning on $\Undist_{u_i}$ can only decrease the probability that we have at least $|R'| \ge \frac{\sqrt{n}}{2}$ many undisturbed roots.
  By applying Lemmas~\ref{lem:undisturbed_roots} and \ref{lem:all_few}, it follows that
  \begin{align}
    \Pr\lt[ (|R'| \!\ge\! \tfrac{\sqrt{n}}{2}) \land \Few_R \land  \Undist_{u_i} \rt]
    \ge
      1 - O \lt( \frac{1}{g(n)} \rt). \label{eq:bnd_fi}
  \end{align}

  To complete the proof, we argue that $u_i$ must send many messages in round $2$, if $\mathcal{F}_i$ happens.
  Let
  \[
  p = \Pr\lt[ m(u_i,2) \!=\! \Omega(n)\ \middle|\ \mathcal{F}_i \rt],
  \]
  and note that our goal is to show that $p = o(1)$.
  By assumption, all nodes in $R'$ are undisturbed in round $1$ of $\Lambda$, and hence the distribution $m(u_i,2)$ is the same for all $u_i \in R'$.
  It follows that the expected message complexity is at least
  \begin{align}
    \EE\lt[ m(2)\ \middle|\ \mathcal{F}_i \rt]
    &\ge
      \EE\lt[ \sum_{u_i \in R'} m(u_i,2)\ \middle|\
        \mathcal{F}_i
      \rt] \notag\\
    &\ge
      \sum_{u_i \in R'}
      \EE\lt[
        m(u_i,2)\ \middle|\
        \mathcal{F}_i
      \rt] \notag\\
    &\ge
      p \cdot
      \sum_{u_i \in R'}
      \EE\lt[
          m(u_i,2)\
        \middle|\
          m(u_i,2) = \Omega(n),
        \mathcal{F}_i
      \rt] \notag\\
     &=
      \Omega \lt( p\,n\,|R'| \rt)\notag\\
     &= \Omega \lt( p\, n^{3/2} \rt). \label{eq:exp_bnd12}
  \end{align}
  It follows that
  \begin{align*}
  \EE\lt[ m(2) \rt]
    &\ge
      \Pr\lt[ \mathcal{F}_i \rt]
      \EE\lt[ m(2) \ \middle|\ \mathcal{F}_i \rt] \\
    \ann{by \eqref{eq:bnd_fi} and \eqref{eq:exp_bnd12}}
    &\ge
      \lt( 1 - \frac{1}{g(n)} \rt)\Omega \lt( p\, n^{3/2} \rt) \notag\\
    &=
      \Omega \lt( p\, n^{3/2} \rt),\notag
  \end{align*}
  and since the algorithm has an assumed message complexity of $o \lt( n^{3/2} \rt)$, this implies $p = o(1)$, as required.
\end{proof}

\begin{lemma} \label{lem:undisturbed_children}
  Let $C_i'$ be the set of undisturbed child nodes of root $u_i$.
  It holds that
  \begin{align*}
    \Pr\lt[
      |C_i'| \ge \tfrac{1}{2}m(u_i,1)
    \ \middle|\
      m(u_i,2) = o(n),
      |R'| \ge \tfrac{\sqrt{n}}{2}, \Few_R, \Undist_{u_i}
    \rt]
    \ge
    1 - O \lt( \tfrac{1}{g(n)} \rt).
  \end{align*}
\end{lemma}
\begin{proof}
  Let $v_{i,j}$ be a child of root $u_i \in R'$ in execution $\Lambda$.
  We first show that
  \begin{align}
    \Pr\lt[ \neg\Undist_{v_{i,j}}
      \ \middle|\
      m(u_i,2) = o(n),
      |R'| \ge \tfrac{\sqrt{n}}{2}, \Few_R, \Undist_{u_i}
    \rt]
    \le
    O\lt(\tfrac{1}{g(n)}\rt). \label{eq:prob_bnd1}
  \end{align}

  Recall that any root node in $R'$ does not receive messages in round $1$. To account for the conditioning on the event $|R'| \ge \frac{1}{2}|R|$,
  we pessimistically assume that $|R'| = |R|$, which means that all messages sent in round $1$ must go to
  $n - |R| = n - \Theta(\sqrt{n}) \ge \tfrac{1}{2}n$
  non-root nodes, one of which is $v_{i,j}$, as this assumption can only decrease the probability of $\Undist_{v_{i,j}}$.
  It follows that the probability that some root sends a round-$1$ message to $v_{i,j}$ is at most $\frac{2\sqrt{n}}{n\,g(n)}$.
  Thus the probability of $\neg\Undist_{v_{i,j}}$ is at most
  \begin{align*}
    &\phantom{\le}\
      1 - \lt( 1 - \frac{2}{\sqrt{n}\,g(n)} \rt)^{|R|} \\
    \ann{since $1 - x \ge e^{-2x}$, for $x<\tfrac{1}{2}$}
    &\le
      1 - \exp \lt( - \frac{4|R|}{\sqrt{n}\,g(n)} \rt) \\
    \ann{since $1 - x \le e^{-x}$}
    &\le
      1 - \lt( 1 - \frac{4|R|}{\sqrt{n}\,g(n)} \rt) \\
    &=
      O\lt(\tfrac{1}{g(n)}\rt),
  \end{align*}
  which proves \eqref{eq:prob_bnd1}.
  It follows that the expected number of disturbed children of $u_i$ is at most
  $m(u_i,1) \lt( 1 - O \lt( \tfrac{1}{g(n)} \rt) \rt),$
  and Markov's Inequality ensures that the probability of
  \[
  m(u_i,1) - |C'| \ge \tfrac{1}{2}m(u_i,1)
  \]
  is at most $O \lt( \tfrac{1}{g(n)} \rt)$.
  Conversely, we have $|C'| \ge \tfrac{1}{2}m(u_i,1)$ with probability at least $1 - O\lt( \tfrac{1}{g(n)} \rt)$.
\end{proof}
}%

%% file: async.tex
In this section, we present a randomized leader election algorithm that achieves the first trade off between messages and time in the asynchronous model.

Indeed, in addition to the adversarial delay of messages, we also allow an adversarial wake up.\footnote{Recall that in the synchronous case we addressed the two kinds of wake up separately---we assumed simultaneous wake-up in Section \ref{sec:sync-clique} since this made our lower bounds harder to show, and addressed the synchronous clique with adversarial wake-up in Section~\ref{sec:adversarial} since this allowed for an even higher lower bound.}
(Note that the lower bound of Section \ref{sec:adversarial} applies here too, since this section deals with adversarial wake up too).
\enlargethispage{\baselineskip}
We assume that the adversary must choose the port mapping \emph{obliviously}, i.e., before the first node is awoken and hence the port connections are independent of the random bits used by the nodes.
However, the adversary may adaptively schedule the steps taken by the nodes: in fact, after fixing the port mappings, we allow the adversary to inspect all random bit strings before deciding which node takes the next step and at what time.
Equivalently to the synchronous model, the \emph{message complexity} of an asynchronous algorithm is the worst-case total number of messages sent during its execution.
Note that we assume that each communication link guarantees FIFO message delivery.
We follow the standard of defining the \emph{asynchronous time complexity} as the worst-case total number of time units since the first node was awoken until the last message was received, whereas a \emph{unit of time} refers to an upper bound on the transmission time of a message.

When a node $u$ wakes up, either by the adversary or due to receiving a message, it samples uniformly at random $\Theta(n^{1/k})$ ports for sending its wake-up message, where $k$ is a parameter that controls the tradeoff between messages and time.
Then, $u$  becomes a \emph{leader candidate} with probability $\Theta(\log n / n)$, and, if it is a candidate, chooses a number called \emph{rank} from a range of a polynomial size that is, a rank that is
unique with high probability.
If $u$ is a candidate, it samples $\Theta(\sqrt{n\log n})$ ports (without replacement) uniformly at random and sends a $\texttt{compete}$ message that carries its rank; the receiving nodes will serve as $u$'s \emph{referees}.
Each node $v$ keeps track of the ``winning'' candidate's rank $\winner$ that
$v$ has seen so far: if $v$ is a candidate, it will store its own rank in its variable $\winner$ (which is empty initially) upon becoming candidate.
Had this been a synchronous network, the referee would have known when all the $\texttt{compete}$  messages arrived, and would have been able to reply \texttt{``you win!''} to the highest ranked node (possibly itself).
Due to the asynchronous arrival of the $\texttt{compete}$ messages, each referee responds \texttt{``you win!''} already to the very first $\texttt{compete}$ message it receives from some node $u$, if $v$ is not a candidate; in that case, $v$ also stores the rank of $u$ in $\winner$.
However, the above ``winnings'' may be revoked later.
That is, whenever $v$ receives a compete message with rank $\rho$ from $u$ when  $\winner$ is not empty, node $v$ immediately replies with a $\texttt{you lose!}$-message if $\rho \le \winner$.
Otherwise,
$v$ consults with the node $w$ whose current rank $v$ stores in $\winner$ to see whether $w$
has already decided to become leader. (Node $w$ may be $v$ itself.)
Assuming that $w$ has not become elected
(for instance,
because $w$ is still waiting to hear back from its other referees), then $w$ itself drops out of the competition upon receiving $v$'s message.
Node $v$, in turn, will consider $u$ to be the current winning candidate, update the value in $\winner$ accordingly, and send a $\texttt{you win!}$ message to $u$.
Eventually, the candidate that receives $\texttt{you win!}$ messages from all its referees and has not yet dropped out of the competition becomes leader.
\onlyLong{We provide the full pseudo code in Algorithm~\ref{alg:async}.}%

\begin{theorem} \label{thm:async}
  For any integer $k$ where $2 \le k \le O\lt( \frac{\log n}{\log\log n} \rt)$, there exists an asynchronous leader election algorithm that, with high probability, elects a unique leader in at most $k+8$ units of time and sends at most $O\lt(n^{1+1/k}\rt)$ messages.
\end{theorem}

\begin{algorithm*}[th]
\begin{algorithmic}[1]
\State Initially, every node is \emph{undecided} and eventually may \emph{decide} to either become leader or non-leader.
\State Each node $u$ has a variable called \emph{winner-rank} $\winner$, which is empty initially, and is used to store the rank of received messages (explained below). In addition, it keeps track of the corresponding node that sent the rank currently stored in winner-rank.
\If{an asleep node $u$ receives a message or is awoken by the adversary}
  \State $u$ sends a message $\msg{\texttt{wake up!}}$ over $\Theta(n^{1/k})$ ports chosen uniformly at random.
\EndIf
\State $u$ becomes a \emph{candidate} with probability $\frac{4\log n}{n}$.
\If{node $u$ is a candidate}
  \State Node $u$ chooses a \emph{rank} $\rho_u$ uniformly at random from $[n^4]$ and stores it in $\winner$.
  \State $u$ samples a set $R_u$ of $\lceil 4\sqrt{n\log n} \rceil$ ports uniformly at random (without replacement).
  \State $u$ sends $\msg{\rho_u,\texttt{compete}}$ over the ports in $R_u$, and considers the nodes in $R_u$ its \emph{referees}.
  \If{$u$ has received $\msg{\texttt{you win!}}$ from all its referees and is still undecided}
    \State $u$ decides to be leader and informs all nodes, who change their statuses to
    non-leader.
    \State \Comment{A node responds to received $\msg{\rho,\texttt{compete}}$-messages even if it has already decided.}
  \EndIf
  \If{$u$ receives a $\msg{\texttt{you lose!}}$ message}
    \State $u$ decides to become non-leader. \label{line:resign1}
  \EndIf
\EndIf
\State Every (decided or undecided) node $v$ processes a received $\msg{\rho_u,\texttt{compete}}$ message as follows:
\If{$v$'s winner-rank $\winner$ is empty}
  \State $v$ stores $\rho_u$ in $\winner$, sends $\msg{\texttt{you win!}}$ to $u$, and $v$ becomes non-leader.
\ElsIf{$\rho_u$ is not larger than $v$'s winner-rank $\winner$}
  \State $v$ sends $\msg{\texttt{you lose!}}$ to $u$, causing $u$ to become non-leader. %
  \label{line:resign2}
\Else \Comment{$\rho_u$ is larger than $v$'s current value of $\winner$}
  \State Let $w$ be the node associated with the value in $\winner$ of $v$.
  \State $v$ asks $w$ whether it is already a leader.
  \If{$w$ has not yet become leader at the time it receives $v$'s request}
    \State $w$ becomes non-leader and informs $v$.
    \State Upon receiving $w$'s answer, $v$ responds to $u$ by sending $\msg{\texttt{you win!}}$.
    \State $v$ stores $u$'s rank in $\winner$.
  \Else
    \State $w$ informs $v$ that it has already become leader.
    \State Upon receiving $w$'s answer, $v$ sends $\msg{\texttt{you lose!}}$ to $u$, causing $u$ to become non-leader. \label{line:resign3}
  \EndIf
\EndIf
\end{algorithmic}
\caption{An asynchronous leader election algorithm that takes a tradeoff parameter $k$, where $2 \le k \le O\lt( {\log n}/{\log\log n} \rt)$.}
\label{alg:async}
\end{algorithm*}

Let us illustrate the tradeoff achieved by Theorem~\ref{thm:async} by considering the two extremes: For $k = 2$, the algorithm terminates in at most $10$ units of time  and sends $O(n^{3/2})$ messages, thus matching the lower bound of Theorem~\ref{thm:adversarial_lb}, whereas for $k=\Theta(\log n / \log\log n)$, we obtain a time complexity of $O(\log n)$ and a message complexity of $O(n\log n)$.
\onlyShort{ }

In the remainder of this section, we prove Theorem~\ref{thm:async}. For the analysis, we consider the execution as being split into two phases: a \emph{wake-up phase}, where the goal is to wake up every node in the network eventually and an \emph{election phase}, in which the nodes determine the leader.
Since the network is asynchronous, it can happen that some nodes enter (or even complete) their election phase, while other nodes are still asleep.

\subsection{Analysis of the Wake-Up Phase} \label{sec:wakeup}

\begin{lemma} \label{lem:wakeup}
  Assume that $2 \le k = O\lt( \frac{\log n}{\log \log n} \rt)$ and
  consider the asynchronous model.
  If each node sends $\gamma n^{1/k}$ messages initially over uniformly at random sampled ports, where $\gamma$ is a sufficiently large constant, then each node wakes up within $k+4$ units of time with high probability.
\end{lemma}

In the remainder of this section, we prove Lemma~\ref{lem:wakeup}.
We call a point in time $t_i$ a \emph{wake-up step of node $u_i$}, if $u$ is awoken by the adversary at time $t_i$, or, because $u$ received a message from another node at time $t_i$.
Without loss of generality, we assume that exactly one node wakes up in each wake-up step and that the adversary wakes up only one node initially.
It is straightforward to verify that our analysis also holds for the case where the adversary wakes up multiple nodes simultaneously, since every node behaves the same way upon wake-up.

\begin{definition}[Covered Nodes] \label{def:covered_nodes}
For any given point in time $t$, we define $C_t$ to be the set of \emph{covered nodes}, which contains every node $u$ that is either awake at time $t$ or, if $u$ is still asleep, then $u$ must be %
the destination of some in-transit message that was sent at some time $t' < t$.
We use $\bar{C}_{t_i}$ to denote the complement of $C_{t_i}$, which contains every node that is still asleep and does not have any in-transit message addressed to it at time $t_i$.
\end{definition}

\begin{lemma} \label{lem:first_phase}
  Suppose that the $i$-th wake-up step occurs at time $t_i$.
  There exists a constant $c\ge 1$ such that, with high probability, for all $i\ge 1$, if $|C_{t_i}| \le \frac{n}{16}$, then,
  \begin{align}
    |C_{t_i+1}| \ge |C_{t_i}| + c\, n^{1/k}. \label{eq:sigma_growth}
  \end{align}
\end{lemma}
\begin{proof}
  Let $u_i$ be the node that woke up at time $t_i$, which means that $u_i$ sends $\gamma n^{1/k}$ messages over uniformly at random chosen ports, where $\gamma$ is a sufficiently large constant.
Fix some enumeration $m_1,\dots,m_{\gamma n^{1/k}}$ of $u_i$'s messages, and, for each $m_i$, define the indicator random variable $Z_i$ such that $Z_i = 1$ if and only if $m_i$ is sent to ${C}_{{t_i}}$. 
Since $u_i$ sends $m_i$ over a port that is distinct from the ports used for $m_1,\dots,m_{i-1}$ and there are $n-1$ ports in total, it follows that 
\begin{align}
 \Pr\big[ \bigwedge_{i = 1}^\ell Z_i = 1\  \big]
  &=
   \prod_{i = 1}^\ell
    \Pr\big[ Z_i = 1\ \big|\  \bigwedge_{j=1}^{i-1} Z_j = 1 \big] \notag\\ 
  &\le
  \prod_{i = 1}^\ell
   \frac{|{C}_{t_i}| - (i-1)}{n-1 -(i-1)} \notag\\ 
	\ann{since $i-1 \le \gamma n^{1/k}$}
  &\le
  \prod_{i = 1}^\ell
   \frac{|{C}_{t_i}|}{n-1 - \gamma n^{1/k}} \notag\\ 
	\ann{since $|{C}_{t_i}| \le n/16$ and $\gamma n^{1/k}+1 \le n/16$}
  &\le
   \lt( \frac{1}{15}  \rt)^{\ell}\label{eq:prob_bnd}
\end{align}
We make use of the following generalized Chernoff bound:
  \begin{lemma}[implicit in Theorem~1.1 in \cite{impagliazzo2010constructive}] \label{lem:chernoff_generalized}
    Let $X_1,\dots,X_N$ be (not necessarily independent) Boolean random variables and suppose that, for some $\delta \in [0,1]$, it holds that, for every index set $S \subseteq [N]$, $\Pr\lt[\bigwedge_{i \in S} X_i\rt] \le \delta^{|S|}$.
    Then, for any $\beta \in [\delta,1]$, we have $\Pr\lt[\sum_{i=1}^N X_i \ge \beta N\rt] \le e^{- 2N(\beta - \delta)^2}$.
  \end{lemma}
According to \eqref{eq:prob_bnd}, we can instantiate Lemma~\ref{lem:chernoff_generalized} for the random variables $Z_1,\dots,Z_{\gamma n^{1/k}}$ by setting $\delta = \frac{1}{15}$
  and $\beta = \tfrac{1}{2}$, which ensures that $\beta - \delta \ge \tfrac{1}{4}$.
  It follows that
  \begin{align*}
    \Pr\Big[ \sum_{i=1}^{\gamma n^{1/k}} Z_i \ge \frac{\gamma}{2} n^{1/k}\Big]
    \le
      e^{-2\gamma(\beta - \delta)^2 n^{1/k}}
    \le
      e^{-\frac{\gamma n^{1/k}}{8}}.
  \end{align*}
  Since $n^{1/k} = \Omega(\log n)$,  it holds with high probability 
  that at least ${\tfrac{\gamma}{2} n^{1/k}}$ nodes in $\bar{C}_{t_i}$ will receive a message from $u_i$ and hence are added to $C_{t_i+1}$.
  Finally, we take a union bound over the at most $O(n)$ wake-up steps for which the premise $|{C}_{t_i}| \le n/16$ holds, which shows that \eqref{eq:sigma_growth} holds for all of them, thereby completing the proof of Lemma~\ref{lem:first_phase}.
\end{proof}

So far, we have shown that the number of nodes that eventually wake up grows by an additive $\Theta(n^{1/k})$ factor at each wake-up step.
We now leverage this result to show that the growth is indeed geometric as expected.
To this end, consider the time $t_*$ of the latest wake-up step such that $|C_{t_*}| \le \frac{n}{16}$ holds.
We define the \emph{cover tree} $\mathcal{T}_{t_*}$ on the nodes in $C_{t_*}$, which has as root the node that was awoken by the adversary and a node $u$ is the parent of node $v$ in $\mathcal{T}_{t_*}$ if and only if $v$ is woken up by a message sent by $u$ at some time $t'<t_*$. 
\begin{lemma} \label{lem:tree}
  Each node in the network appears at most once in $\mathcal{T}_{t_*}$, and $V(\mathcal{T}_{t_*})=C_{t_*}$. The edges of $\mathcal{T}_{t_*}$ form a directed tree where each non-leaf node has at least $c_1 n^{1/k}$ and at most $\gamma n^{1/k}$ children.
Moreover, any node that has not yet awoken by time $t_*$, but to which a message was sent before time $t_*$, is a leaf in $\mathcal{T}_{t_*}$.
\end{lemma}
\begin{proof}
  By definition, each node in $\mathcal{T}_{t_*}$ (except the root) is awoken due to a message sent at some time before $t_*$, and linearity of time ensures that the tree is acyclic.
	Thus, every node in the network shows up at most once in $\mathcal{T}_{t_*}$.
  Moreover, if a node $v$ is a leaf, then it must be woken up by a message that was sent but not yet received before time $t_*$, as otherwise $v$ itself would have added (at least) $c_1 n^{1/k}$ children to $\mathcal{T}_{t_*}$ according to Lemma~\ref{lem:first_phase}.
From Definition~\ref{def:covered_nodes} is straightforward to verify that a node is in $\mathcal{T}_{t_*}$ if and only if it is in $C_{t_*}$.
  The upper bound on the degree is immediate from the fact that a node sends $\gamma n^{1/k}$ messages upon wake-up, whereas the lower bound follows by applying Lemma~\ref{lem:first_phase} to each non-leaf node in $\mathcal{T}$.
\end{proof}

\begin{lemma} \label{lem:wakeup_tree_depth}
  With high probability, every node in the tree $\mathcal{T}_{t_*}$ is woken up by a messages that was sent before time $k+2$.
\end{lemma}
\begin{proof}
  Let $u_1$ be the root of the tree. %
  Assume towards a contradiction that there exists a path $p = u_1,u_2,\dots,u_\ell$ in $\mathcal{T}_{t_*}$ such that the message that wakes up the leaf $u_\ell$ is sent at time $t_\ell \ge k+2$,
  and suppose that path $p$ is chosen such that $t_\ell$ is the latest time for all paths satisfying the property. 
  Since each message takes at most $1$ unit of time, it follows that $\ell \ge k+3$.
  (Note that we assume that the root $u_1$ is awoken at time $0$.)
  Now consider any path $q$ starting from the root $u_1$ and continuing over a distinct child $u_2' \ne u_2$.
  Let $u_m'$ be the last node on path $q$, i.e., $u_m'$ is a leaf.

  We first argue  that the message that wakes up $u_m'$ must be sent at some time \[t' \ge t_\ell - 1 \ge k+1.\]
  Recall that $\mathcal{T}_{t_*}$ contains all nodes that are either awake or to which a message is in-transit at any time before $t_*$, which, in particular, includes time $t_\ell$, and assume towards a contradiction that $t' < t_\ell - 1$.
  By assumption, $u_m'$ must wake up by time 
  \[t'+1 < t_\ell \le t_*.\]
  Since $t_\ell \le t_*$, it follows that $|C_{t_\ell}| \le |C_{t_*}| \le \tfrac{n}{16}$, and hence Lemma~\ref{lem:first_phase} ensures that $u_m'$ sends messages to previously asleep nodes by time $t' +1 < t_*$, contradicting the assumption that $u_m'$ is a leaf.

  Thus, we can assume that $t' \ge t_\ell - 1 \ge k+1$, which implies that $q$ contains at least $k+2$ nodes, and $k+1$ of them must be non-leaf nodes.
  Therefore, any path from the root consists of at least $k+1$ non-leaf nodes that have a minimum degree of $c_1 n^{1/k}$, according to Lemma~\ref{lem:tree}.
  Overall, this means that there are at least $c_1^{k+1}\,n^{1+1/k} > n$ nodes in $\mathcal{T}_{t_*}$, since $c_1 \ge 1$, thus providing a contradiction.
\end{proof}

To complete the proof of Lemma~\ref{lem:wakeup}, we show that, once we have reached a ``critical mass'' of $\frac{n}{16}$ nodes that are either awake (or about to wake up), the remaining nodes wake up within $2$ additional time steps, resulting in a time complexity of $t_* + 2 \le k+4$, as claimed.

\begin{lemma} \label{lem:final_step}
  With high probability, all nodes wake up by time $t_* + 2$.
\end{lemma}
\begin{proof}
  Recall the cover tree $\mathcal{T}_{t_*}$ for time $t_*$ defined above, which was the latest wake-up step at which $|C_{t_*}| \le \frac{n}{16}$.
	Since
    $|C_{t_*}| \ge \frac{n}{16} - 1 \ge \frac{n}{32}$,
  it follows that $\mathcal{T}_{t_*}$ contains at least $|\mathcal{T}_{t_*}|-\frac{|\mathcal{T}_{t_*}|}{c_1n^{1/k}}\ge\frac{n}{32} - \frac{n^{1 - 1/k}}{c_1} \ge \frac{n}{64}$ leafs.
  According to Lemma~\ref{lem:tree}, a leaf is a node that has not yet woken up by time $t_*$, which means that there are at least $\frac{n}{64}$ nodes that will wake up  due to messages sent before time $t_*$, all of which must arrive at some point in the time interval $[t_*,t_* + 1]$.

  Let $M$ be the set of the messages sent by these nodes during their wake-up steps and note that $|M| \ge \frac{\gamma n^{1+1/k}}{64}$.
  The probability that a node $u \in \bar{C}_{t_*}$ is \emph{not} the destination of at least one of the messages in $M$ is at most
  \[
    \lt( 1 - \frac{1}{n} \rt)^{{\gamma n^{1+1/k}}/{64}}
    \le
    e^{\lt( -\frac{\gamma n^{1+1/k}}{64}\rt)}
    \le
    n^{-4},
  \]
  where the last inequality follows from the assumption that $n^{1/k} = \Omega(\log n)$ and that $\gamma$ is a sufficiently large constant.
  Observing that all messages in $M$ must arrive by time $t_*+2$, and taking a union bound over the at most $\frac{31}{32}n$ nodes that are not in $C_{t_*}$, ensures that every node will be awoken by time $t_* + 2$ with high probability.
\end{proof}

\subsection{Analysis of the Election Phase} \label{sec:election}

\begin{lemma} \label{lem:one_leader}
  Suppose that every node wakes up eventually.
  Then there is exactly one leader with high probability.
\end{lemma}
\begin{proof}
  Since each node becomes candidate with probability $4\log n / n$, it follows that the probability there is a candidate is at least $1 - \lt(1 - \frac{4\log n}{n}\rt)^n \ge 1 - e^{-4\log n} = 1 - \frac{1}{n^4}$.
  By a simple calculation that is analogous to the one in the proof of Theorem~\ref{thm:adversarial_two_round_algo} in Section~\ref{sec:adversarial}, it follows that every candidate will have a unique ID with high probability.
  For the remainder of the proof, we condition on these high-probability events.

  Now, assume towards a contradiction that there exist two candidates $u$ and $v$ that both become leaders.
  Let $R_u$ and $R_v$ be the set of referees chosen by $u$ respectively $v$.
  Since it holds that
  \begin{align*}
    \Pr[ R_u \cap R_v = \emptyset ]
    \le \lt(1 - \frac{4\sqrt{\log n}}{\sqrt{n}}\rt)^{4\sqrt{n\log n}}
    \le \frac{1}{n^{2}},
  \end{align*}
  we know that $u$ and $v$ will share some common referee $w$.
  By the code of the algorithm, we know that $w$ sends a winner-message at most once to each node that contacts it.
  By the code of the algorithm, any undecided node becomes non-leader upon receiving a $\msg{\texttt{you lose!}}$ message, which tells us that
  $w$ sent a message $\msg{\texttt{you win!}}$ to both $u$ and $v$.
  Without loss of generality, assume that $w$ sent the $\msg{\texttt{you win!}}$ first to $u$ and later on to $v$.
  It follows that $w$ must have processed the message $\msg{\rho_u}$ before $\msg{\rho_v}$.
  However, by the code of the algorithm, $w$ stores $u$'s rank in its winner-rank variable upon sending $\msg{\texttt{you win!}}$ to $u$ at some time $t_1$.
  Moreover, since $w$ sent a $\msg{\texttt{you win!}}$ also to $v$ at some later time $t_2>t_1$, it follows that $u$ cannot have elected itself as a leader before $t_2$, as otherwise $w$ would have sent $\msg{\texttt{you lose!}}$ to $v$ upon receiving $u$'s reply.
  This implies that $u$ must have already decided to become non-leader upon being contacted by $v$.
  We have arrived at a contradiction.
\end{proof}

\begin{lemma} \label{lem:election_complexity}
  Each node decides in at most $4$ units of time after its wake-up step.
\end{lemma}
\begin{proof}
  Suppose node $u$ wakes up at time $t$.
  If $u$ does not become a candidate, it decides instantly; thus, assume that $u$ does become a candidate and contacts a set of referees.
  Upon receiving $u$'s message, a referee may need to confirm the status of the current winner, which takes one additional round-trip before it sends its response to $u$.
  Thus, it takes at most $4$ units of time until $u$ receives a response from each of its referees and decides.
  Note that $u$ may also serve as referee for other nodes, however, responding to these messages does not slow down its own decision process.
\end{proof}

\subsection{Putting Everything Together}

We have now assembled all the ingredients for completing the proof of Theorem~\ref{thm:async}.
We first argue correctness:
With high probability, the algorithm will elect exactly one node as the leader since Lemma~\ref{lem:wakeup} shows that every node wakes up and, in that case, Lemma~\ref{lem:one_leader} guarantees that exactly one of them will become leader.

For the claimed time complexity of $k+8$, we observe that Lemma~\ref{lem:wakeup} ensures that every node wakes up within $k+4$ units of time, and Lemma~\ref{lem:election_complexity} tells us that everyone decides to become leader or non-leader within $4$ additional time units.

Finally, note that each node sends $\Theta(n^{1+1/k})$ messages upon wake-up.
We have $\Theta(\log n)$ candidates (whp), and each of them communicates with $\Theta(\sqrt{n\log n})$ referees.
Upon being contacted by a candidate, a referee may need to communicate with its currently stored winner node, which requires an additional $2$ messages.
Thus, we have an overall message complexity of $O\lt( n^{1+1/k} \rt) + O\lt(\sqrt{n}\log^{3/2}n\rt) = O\lt( n^{1+1/k} \rt)$ with high probability.

%% file: asynchronous_afek_and_gafni.tex
\onlyLong{
The following algorithm can be viewed as an
asynchronous version of the synchronous algorithm of
 Afek and Gafni~\cite{afek1991time}. However, Afek and Gafni
argue that their algorithm would not work in asynchronous networks, and thus suggested three asynchronous algorithms where the $O(n \log n)$
messages were obtained at the cost of $\Omega(n)$ time. We managed to make their original algorithm asynchronous without increasing their original time, assuming that we count only from the last spontaneous (adversarial) wake up of a node. Alternatively, for the sake of this algorithm only, we can assume that all the nodes wake up simultaneously.
}%
\onlyShort{
In the full paper~\cite{DBLP:journals/corr/abs-2301-08235}, we give an algorithm that can be viewed as an
asynchronous version of the synchronous algorithm of
 Afek and Gafni~\cite{afek1991time}. However, Afek and Gafni
argue that their algorithm would not work in asynchronous networks, and thus suggested three asynchronous algorithms where the $O(n \log n)$
messages were obtained at the cost of $\Omega(n)$ time. We managed to make their original algorithm asynchronous without increasing their original time, assuming that we count only from the last spontaneous (adversarial) wake up of a node. Alternatively, for the sake of this algorithm only, we can assume that all the nodes wake up simultaneously.
}\xspace
\onlyLong{%
The algorithm works as follows: Each surviving candidate $v$ (initially everybody) is at a level in $[0,1,2,... \log n]$ (initially zero). It sends $2^i$ request messages to its first $2^i$ neighbors ($v$ is its own neighbor number 1; we also assume that the ports are numbered $2,3,...n$). If $v$ gets acks for all the messages it sends (and has not meanwhile been killed) then it climbs to level $i+1$. If it sent messages to all the nodes and received acks, then it terminates as a leader.

A node $v$ that gets a request message for level $i$ and has not sent an ack yet, now sends an ack. If $v$ did send an ack to some $u$ and now receives a request in (level i) from $w>u$ (w.r.t.\ to their IDs)  then $v$ sends a conditional cancel message to $u$.
If $u$ is already in a level higher than i, then $u$ refuses, and $v$ kills $w$. Otherwise, $u$ is killed and $v$ sends an ack to $w$ (unless $v$ meanwhile received a request from some $z>w$).

\subsubsection*{Analysis}

\begin{lemma}
If $2^i<n$ and there was a (live) candidate at level
$i$, then there is a candidate at level $i+1$.
\end{lemma}

\begin{proof}
Consider the highest candidate $w$ to reach level $i$. Suppose that $w$ never reaches level $i+1$. Since $w$ is the highest,
the only way it can get killed is if it sent a request to some node $v$, and $v$ sent as a result a cancel message to $w$ because some node $u$ (with a smaller ID than $w$) already reached level $i+1$. 
This proves the lemma.
\end{proof}

\begin{lemma}
No more than $\frac{n}{2^i}$ candidates reach level $i$.
    \end{lemma}

\begin{proof}
Consider a candidate $v$ who reached level $i+1$. Then $v$ received an ack from $2^i$ nodes. Among above mentioned $2^i$ nodes, consider some node $u$. If $u$ had sent an ack to some other node $w$ before sending an ack to $v$, then $u$ had also sent a cancel message to $w$ and has verified that $w$ stopped being a candidate, before $u$ sent an ack to $v$. Moreover, if $u$ has not sent an ack to any node after sending to $v$, since otherwise, this would have meant that $v$ accepted a cancel request and ceased being a candidate before reaching level $i+1$. Hence, for every candidate $v$ who reached level $i+1$ there exist a set of $2^i$ nodes who supported $v$ at the time it increased its level to $i+1$, and none of them supported any other node $w$ when $w$ increased its own level to $i+1$. The lemma follows.
\end{proof}

\begin{corollary}
The number of levels is at most $O(\log n)$, the total number of messages sent by nodes on each level is $O(n)$, and the time duration of each level is constant.
\end{corollary}

The above lemmas and corollary establish the following theorem:
}

\begin{theorem}
There exist a leader election algorithm for asynchronous networks under simultaneous wake-up that uses $O(\log n)$ time and $O(n\log n)$ messages, if the time is counted from the last spontaneous (adversarial) wake up of a node.
\end{theorem}

%% file: main.bbl

\begin{thebibliography}{23}


\ifx \showCODEN    \undefined \def \showCODEN     #1{\unskip}     \fi
\ifx \showDOI      \undefined \def \showDOI       #1{#1}\fi
\ifx \showISBNx    \undefined \def \showISBNx     #1{\unskip}     \fi
\ifx \showISBNxiii \undefined \def \showISBNxiii  #1{\unskip}     \fi
\ifx \showISSN     \undefined \def \showISSN      #1{\unskip}     \fi
\ifx \showLCCN     \undefined \def \showLCCN      #1{\unskip}     \fi
\ifx \shownote     \undefined \def \shownote      #1{#1}          \fi
\ifx \showarticletitle \undefined \def \showarticletitle #1{#1}   \fi
\ifx \showURL      \undefined \def \showURL       {\relax}        \fi
\providecommand\bibfield[2]{#2}
\providecommand\bibinfo[2]{#2}
\providecommand\natexlab[1]{#1}
\providecommand\showeprint[2][]{arXiv:#2}

\bibitem[Afek and Gafni(1991)]%
        {afek1991time}
\bibfield{author}{\bibinfo{person}{Yehuda Afek} {and} \bibinfo{person}{Eli
  Gafni}.} \bibinfo{year}{1991}\natexlab{}.
\newblock \showarticletitle{Time and message bounds for election in synchronous
  and asynchronous complete networks}.
\newblock \bibinfo{journal}{\emph{SIAM J. Comput.}} \bibinfo{volume}{20},
  \bibinfo{number}{2} (\bibinfo{year}{1991}), \bibinfo{pages}{376--394}.
\newblock


\bibitem[Afek and Matias(1994)]%
        {afek1994elections}
\bibfield{author}{\bibinfo{person}{Yehuda Afek} {and} \bibinfo{person}{Yossi
  Matias}.} \bibinfo{year}{1994}\natexlab{}.
\newblock \showarticletitle{Elections in anonymous networks}.
\newblock \bibinfo{journal}{\emph{Information and Computation}}
  \bibinfo{volume}{113}, \bibinfo{number}{2} (\bibinfo{year}{1994}),
  \bibinfo{pages}{312--330}.
\newblock


\bibitem[Angluin et~al\mbox{.}(2004)]%
        {popul}
\bibfield{author}{\bibinfo{person}{Dana Angluin}, \bibinfo{person}{James
  Aspnes}, \bibinfo{person}{Zo{\"e} Diamadi}, \bibinfo{person}{Michael~J
  Fischer}, {and} \bibinfo{person}{Ren{\'e} Peralta}.}
  \bibinfo{year}{2004}\natexlab{}.
\newblock \showarticletitle{Computation in networks of passively mobile
  finite-state sensors}. In \bibinfo{booktitle}{\emph{Proceedings of the
  twenty-third annual ACM symposium on Principles of distributed computing}}.
  \bibinfo{pages}{290--299}.
\newblock


\bibitem[Awerbuch et~al\mbox{.}(1990)]%
        {awerbuch1990trade}
\bibfield{author}{\bibinfo{person}{Baruch Awerbuch}, \bibinfo{person}{Oded
  Goldreich}, \bibinfo{person}{Ronen Vainish}, {and} \bibinfo{person}{David
  Peleg}.} \bibinfo{year}{1990}\natexlab{}.
\newblock \showarticletitle{A trade-off between information and communication
  in broadcast protocols}.
\newblock \bibinfo{journal}{\emph{Journal of the ACM (JACM)}}
  \bibinfo{volume}{37}, \bibinfo{number}{2} (\bibinfo{year}{1990}),
  \bibinfo{pages}{238--256}.
\newblock


\bibitem[Derakhshandeh et~al\mbox{.}(2015)]%
        {derakhshandeh2015leader}
\bibfield{author}{\bibinfo{person}{Zahra Derakhshandeh},
  \bibinfo{person}{Robert Gmyr}, \bibinfo{person}{Thim Strothmann},
  \bibinfo{person}{Rida Bazzi}, \bibinfo{person}{Andr{\'e}a~W Richa}, {and}
  \bibinfo{person}{Christian Scheideler}.} \bibinfo{year}{2015}\natexlab{}.
\newblock \showarticletitle{Leader election and shape formation with
  self-organizing programmable matter}. In
  \bibinfo{booktitle}{\emph{International Workshop on DNA-Based Computers}}.
  Springer, \bibinfo{pages}{117--132}.
\newblock


\bibitem[Frederickson(1983)]%
        {frederickson1983tradeoffs}
\bibfield{author}{\bibinfo{person}{Greg~N Frederickson}.}
  \bibinfo{year}{1983}\natexlab{}.
\newblock \showarticletitle{Tradeoffs for selection in distributed networks
  (Preliminary Version)}. In \bibinfo{booktitle}{\emph{Proceedings of the
  second annual ACM symposium on Principles of distributed computing}}.
  \bibinfo{pages}{154--160}.
\newblock


\bibitem[Frederickson and Lynch(1987)]%
        {frederickson1987electing}
\bibfield{author}{\bibinfo{person}{Greg~N Frederickson} {and}
  \bibinfo{person}{Nancy~A Lynch}.} \bibinfo{year}{1987}\natexlab{}.
\newblock \showarticletitle{Electing a leader in a synchronous ring}.
\newblock \bibinfo{journal}{\emph{Journal of the ACM (JACM)}}
  \bibinfo{volume}{34}, \bibinfo{number}{1} (\bibinfo{year}{1987}),
  \bibinfo{pages}{98--115}.
\newblock


\bibitem[Ganeriwal et~al\mbox{.}(2003)]%
        {ganeriwal2003timing}
\bibfield{author}{\bibinfo{person}{Saurabh Ganeriwal}, \bibinfo{person}{Ram
  Kumar}, {and} \bibinfo{person}{Mani~B Srivastava}.}
  \bibinfo{year}{2003}\natexlab{}.
\newblock \showarticletitle{Timing-sync protocol for sensor networks}. In
  \bibinfo{booktitle}{\emph{Proceedings of the 1st international conference on
  Embedded networked sensor systems}}. \bibinfo{pages}{138--149}.
\newblock


\bibitem[Humblet(1984)]%
        {humblet1984selecting}
\bibfield{author}{\bibinfo{person}{Pierre~A Humblet}.}
  \bibinfo{year}{1984}\natexlab{}.
\newblock \bibinfo{booktitle}{\emph{Selecting a leader in a clique in 0 (N log
  N) messages}}.
\newblock \bibinfo{publisher}{Laboratory for Information and Decision Systems,
  Massachusetts Institute of~…}.
\newblock


\bibitem[Impagliazzo and Kabanets(2010)]%
        {impagliazzo2010constructive}
\bibfield{author}{\bibinfo{person}{Russell Impagliazzo} {and}
  \bibinfo{person}{Valentine Kabanets}.} \bibinfo{year}{2010}\natexlab{}.
\newblock \showarticletitle{Constructive proofs of concentration bounds}.
\newblock In \bibinfo{booktitle}{\emph{Approximation, Randomization, and
  Combinatorial Optimization. Algorithms and Techniques}}.
  \bibinfo{publisher}{Springer}, \bibinfo{pages}{617--631}.
\newblock


\bibitem[King et~al\mbox{.}(2015)]%
        {KKT}
\bibfield{author}{\bibinfo{person}{Valerie King}, \bibinfo{person}{Shay
  Kutten}, {and} \bibinfo{person}{Mikkel Thorup}.}
  \bibinfo{year}{2015}\natexlab{}.
\newblock \showarticletitle{Construction and impromptu repair of an MST in a
  distributed network with o (m) communication}. In
  \bibinfo{booktitle}{\emph{Proceedings of the 2015 ACM Symposium on Principles
  of Distributed Computing}}. \bibinfo{pages}{71--80}.
\newblock


\bibitem[Korach et~al\mbox{.}(1990)]%
        {korach1990modular}
\bibfield{author}{\bibinfo{person}{Ephraim Korach}, \bibinfo{person}{Shay
  Kutten}, {and} \bibinfo{person}{Shlomo Moran}.}
  \bibinfo{year}{1990}\natexlab{}.
\newblock \showarticletitle{A modular technique for the design of efficient
  distributed leader finding algorithms}.
\newblock \bibinfo{journal}{\emph{ACM Transactions on Programming Languages and
  Systems (TOPLAS)}} \bibinfo{volume}{12}, \bibinfo{number}{1}
  (\bibinfo{year}{1990}), \bibinfo{pages}{84--101}.
\newblock


\bibitem[Korach et~al\mbox{.}(1984)]%
        {korach1984tight}
\bibfield{author}{\bibinfo{person}{Ephraim Korach}, \bibinfo{person}{Shlomo
  Moran}, {and} \bibinfo{person}{Shmuel Zaks}.}
  \bibinfo{year}{1984}\natexlab{}.
\newblock \showarticletitle{Tight lower and upper bounds for some distributed
  algorithms for a complete network of processors}. In
  \bibinfo{booktitle}{\emph{Proceedings of the third annual ACM symposium on
  Principles of distributed computing}}. \bibinfo{pages}{199--207}.
\newblock


\bibitem[Kutten et~al\mbox{.}(2020)]%
        {kutten2020singularly}
\bibfield{author}{\bibinfo{person}{Shay Kutten}, \bibinfo{person}{William~K
  Moses~Jr}, \bibinfo{person}{Gopal Pandurangan}, {and} \bibinfo{person}{David
  Peleg}.} \bibinfo{year}{2020}\natexlab{}.
\newblock \showarticletitle{Singularly Optimal Randomized Leader Election}. In
  \bibinfo{booktitle}{\emph{34th International Symposium on Distributed
  Computing}}.
\newblock


\bibitem[Kutten et~al\mbox{.}(2015a)]%
        {kutten2015complexity}
\bibfield{author}{\bibinfo{person}{Shay Kutten}, \bibinfo{person}{Gopal
  Pandurangan}, \bibinfo{person}{David Peleg}, \bibinfo{person}{Peter
  Robinson}, {and} \bibinfo{person}{Amitabh Trehan}.}
  \bibinfo{year}{2015}\natexlab{a}.
\newblock \showarticletitle{On the complexity of universal leader election}.
\newblock \bibinfo{journal}{\emph{Journal of the ACM (JACM)}}
  \bibinfo{volume}{62}, \bibinfo{number}{1} (\bibinfo{year}{2015}),
  \bibinfo{pages}{1--27}.
\newblock


\bibitem[Kutten et~al\mbox{.}(2015b)]%
        {kutten2015sublinear}
\bibfield{author}{\bibinfo{person}{Shay Kutten}, \bibinfo{person}{Gopal
  Pandurangan}, \bibinfo{person}{David Peleg}, \bibinfo{person}{Peter
  Robinson}, {and} \bibinfo{person}{Amitabh Trehan}.}
  \bibinfo{year}{2015}\natexlab{b}.
\newblock \showarticletitle{Sublinear bounds for randomized leader election}.
\newblock \bibinfo{journal}{\emph{Theoretical Computer Science}}
  \bibinfo{volume}{561} (\bibinfo{year}{2015}), \bibinfo{pages}{134--143}.
\newblock


\bibitem[Le~Lann(1977)]%
        {le1977distributed}
\bibfield{author}{\bibinfo{person}{G{\'e}rard Le~Lann}.}
  \bibinfo{year}{1977}\natexlab{}.
\newblock \showarticletitle{Distributed Systems-Towards a Formal Approach.}. In
  \bibinfo{booktitle}{\emph{IFIP congress}}, Vol.~\bibinfo{volume}{7}. Toronto,
  \bibinfo{pages}{155--160}.
\newblock


\bibitem[Lynch(1996)]%
        {lynch1996distributed}
\bibfield{author}{\bibinfo{person}{Nancy~A Lynch}.}
  \bibinfo{year}{1996}\natexlab{}.
\newblock \bibinfo{booktitle}{\emph{Distributed algorithms}}.
\newblock \bibinfo{publisher}{Elsevier}.
\newblock


\bibitem[Nygren et~al\mbox{.}(2010)]%
        {nygren2010akamai}
\bibfield{author}{\bibinfo{person}{Erik Nygren}, \bibinfo{person}{Ramesh~K
  Sitaraman}, {and} \bibinfo{person}{Jennifer Sun}.}
  \bibinfo{year}{2010}\natexlab{}.
\newblock \showarticletitle{The akamai network: a platform for high-performance
  internet applications}.
\newblock \bibinfo{journal}{\emph{ACM SIGOPS Operating Systems Review}}
  \bibinfo{volume}{44}, \bibinfo{number}{3} (\bibinfo{year}{2010}),
  \bibinfo{pages}{2--19}.
\newblock


\bibitem[Peleg(2000)]%
        {peleg2000distributed}
\bibfield{author}{\bibinfo{person}{David Peleg}.}
  \bibinfo{year}{2000}\natexlab{}.
\newblock \bibinfo{booktitle}{\emph{Distributed computing: a locality-sensitive
  approach}}.
\newblock \bibinfo{publisher}{SIAM}.
\newblock


\bibitem[Ramanathan et~al\mbox{.}(2007)]%
        {ramanathan2007randomized}
\bibfield{author}{\bibinfo{person}{Murali~Krishna Ramanathan},
  \bibinfo{person}{Ronaldo~A Ferreira}, \bibinfo{person}{Suresh Jagannathan},
  \bibinfo{person}{Ananth Grama}, {and} \bibinfo{person}{Wojciech
  Szpankowski}.} \bibinfo{year}{2007}\natexlab{}.
\newblock \showarticletitle{Randomized leader election}.
\newblock \bibinfo{journal}{\emph{Distributed Computing}} \bibinfo{volume}{19},
  \bibinfo{number}{5} (\bibinfo{year}{2007}), \bibinfo{pages}{403--418}.
\newblock


\bibitem[Singh(1992)]%
        {singh1992leader}
\bibfield{author}{\bibinfo{person}{Gurdip Singh}.}
  \bibinfo{year}{1992}\natexlab{}.
\newblock \showarticletitle{Leader election in complete networks}. In
  \bibinfo{booktitle}{\emph{Proceedings of the eleventh annual ACM symposium on
  Principles of distributed computing}}. \bibinfo{pages}{179--190}.
\newblock


\bibitem[Yao and Gehrke(2002)]%
        {yao2002cougar}
\bibfield{author}{\bibinfo{person}{Yong Yao} {and} \bibinfo{person}{Johannes
  Gehrke}.} \bibinfo{year}{2002}\natexlab{}.
\newblock \showarticletitle{The cougar approach to in-network query processing
  in sensor networks}.
\newblock \bibinfo{journal}{\emph{ACM Sigmod record}} \bibinfo{volume}{31},
  \bibinfo{number}{3} (\bibinfo{year}{2002}), \bibinfo{pages}{9--18}.
\newblock


\end{thebibliography}
